\documentclass[runningheads]{llncs}
\usepackage[T1]{fontenc}
\usepackage{microtype}
\usepackage{graphicx} 
\usepackage[caption=false]{subfig}
\usepackage{hyperref}
\usepackage{xfrac}
\usepackage{amssymb}  
\usepackage{mathtools} 
\usepackage{esvect,xspace}
\usepackage{todonotes} 
\usepackage{comment}
\usepackage{cite}
\usepackage{mdframed}
\usepackage[capitalize,nameinlink,noabbrev]{cleveref}
\usepackage{tabularx} 

\usepackage{thmtools, thm-restate}
\usepackage{xpatch}
\usepackage{apptools}

\declaretheorem[name=Observation]{observation}
\newcommand{\obsref}[1]{%
  \hyperref[#1]{Observation~\ref*{#1}}%
}

\newbool{inappendix}
\newbool{inintro}

\NewDocumentCommand{\smartstar}{m}{%
\edef\temp{%
\noexpand\ifbool{inappendix}{%
\noexpand\hyperref[#1]{\appsymb}%
}{\noexpand\hyperref[proof:#1]{\appsymb}%
}%
}%
\temp%
}%

\newif\iflong
\newif\ifshort

\longtrue

\iflong
\else
\shorttrue
\fi

\usepackage{etoolbox}
\newcommand{\appsymb}{$\bigstar$}

\newcommand{\toappendix}[1]{%
  \iflong{}#1\else{}
    \gappto{\appendixText}
    {
        #1
      }
  \fi{}%
}

\newcommand{\notappendixproof}[3]{#3}

\newcommand{\executecommand}[1]{%
    \csname #1\endcsname
}

\newcommand{\appendixproof}[3]{%
  \iflong{}#3\else{}\gappto{\appendixText}
    {
      \subsection{Proof of \cref{#1}}\label{proof:#1}
      \ifshort \csname#2\endcsname*\fi
      #3
    }
  \fi{}
}

\def\hmargin{0.0em}
\newcommand{\problembox}[3]{
  \par\addvspace{\abovedisplayskip}
  \noindent
  \makebox[\textwidth]{%
    \fboxsep=0.5em
    \fbox{%
      \begin{minipage}{\dimexpr\textwidth-2\fboxsep-2\fboxrule-\hmargin\relax}
        \textsc{#1}\par
        \smallskip
        \begin{tabularx}{\textwidth}{@{}l@{\hspace{0.5em}}X@{}}
          \textbf{Input:} & #2 \\
          \textbf{Question:} & #3
        \end{tabularx}
        \vspace{-0.33em}
      \end{minipage}%
    }%
  }
  \par\addvspace{\belowdisplayskip}
}

\crefname{observation}{observation}{observations}
\Crefname{observation}{Observation}{Observations}

\newcommand{\probname}[1]{{\normalfont\textsc{#1}}}

\newcommand{\Size}[1]{\ensuremath{\left\vert #1 \right\vert}}
\newcommand{\Set}[1]{\ensuremath{\left\{ #1 \right\} }}
\newcommand{\NP}{\textsf{\textup{NP}}}
\newcommand{\paraNP}{\textsf{\textup{paraNP}}}
\newcommand{\FPT}{\textsf{\textup{FPT}}}
\newcommand{\XP}{\textsf{\textup{XP}}}
\newcommand{\ETR}{\textsf{\textup{ETR}}}
\newcommand{\W}[1]{\textsf{\textup{W[{#1}]}}}
\newcommand{\BigO}[1]{\ensuremath{\mathcal{O}(#1)}}

\def\fenkernelsizeexpr{10k \cdot 81^k}

\newcommand{\drawing}[1]{geometric $#1$-layer drawing}

\newcommand{\convexhull}[0]{\ensuremath{\text{ConvexHull}}}
\newcommand{\disk}[2]{\ensuremath{\text{Disk}(#1, #2)}}

\newcommand{\cupdot}{\mathbin{\mathaccent\cdot\cup}}

\newcommand{\GT}{\textup{\textsc{GT}}\xspace}
\newcommand{\GTfull}{\textup{\textsc{Geometric Thickness}}\xspace}
\newcommand{\GTE}{\textup{\textsc{GTE}}\xspace}
\newcommand{\GTEfull}{\textup{\textsc{Geometric Thickness Extension}}\xspace}

\begin{document}
\title{Pathways to Tractability for Geometric Thickness\thanks{All authors acknowledge support from the Vienna Science and Technology Fund (WWTF)
[10.47379/ICT22029]. Robert Ganian and Alexander Firbas also acknowledge support from the Austrian Science Fund (FWF) [10.55776/Y1329].}}

\author{Thomas Depian \and
Simon Dominik Fink  \and
Alexander Firbas  \and
Robert Ganian \and
Martin	N\"ollenburg
}

\authorrunning{T. Depian  et al.}
\institute{TU Wien, Vienna, Austria\\
\email{\{tdepian,sfink,afirbas,rganian,noellenburg\}@ac.tuwien.ac.at}}
\maketitle              %
\begin{abstract}
We study the classical problem of computing geometric thickness, i.e., finding a straight-line drawing of an input graph and a partition of its edges into as few parts as possible so that each part is crossing-free. Since the problem is \NP-hard, we investigate its tractability through the lens of parameterized complexity. 
As our first set of contributions, we provide two fixed-parameter algorithms which utilize well-studied parameters of the input graph, notably the vertex cover and feedback edge numbers. Since parameterizing by the thickness itself does not yield tractability and the use of other structural parameters remains open due to general challenges identified in previous works, as our second set of contributions, we propose a different pathway to  tractability for the problem: extension of partial solutions. In particular, we establish a full characterization of the problem's parameterized complexity in the extension setting depending on whether we parameterize by the number of missing vertices, edges, or both. 
\end{abstract}

\setbool{inintro}{true}
\section{Introduction}
The \emph{thickness} of a graph $G$ is the minimum integer $\ell$ such that there exists a drawing $\Gamma$ of $G$ and a partitioning of its edges into $\ell$ layers such that no two edges assigned to the same layer cross in $\Gamma$.
Thickness is a classical generalization of graph planarity;
its early study dates back to the seventies~\cite{tutte1963thickness,beineke1965thickness} and yet it remains a prominent topic of research to this day~\cite{DurocherGM16,DurocherM18,CheongPS22,0015RR023}.
An overview of classical results on thickness and its initial applications can be found, e.g., in the dedicated survey of Mutzel, Odenthal, and Scharbrodt~\cite{MutzelOS98}.

While the distinction of whether $\Gamma$ represents edges as curves or line segments is immaterial when determining whether $G$ is planar (i.e., whether it has thickness $1$), this is far from true for higher thickness values. In fact, a seminal paper of Eppstein established that the gap between \emph{geometric thickness} (with edges represented as straight-line segments) and \emph{graph-theoretic thickness} (where edges are curves) can be arbitrarily large~\cite{Eppstein02}. Both of these notions are now often studied independently, with a number of articles focusing on various combinatorial properties of geometric or graph-theoretic thickness~\cite{DujmovicW07,DurocherGM16,DurocherM18,0015RR023}.

\looseness=-1
In terms of computational complexity, determining whether a graph has graph-theoretic thickness $2$ was shown to be \NP-hard already 40 years ago~\cite{mansfield1983determining}. The analogous question for geometric thickness was resolved only in 2016~\cite{DurocherGM16}---16 years after it was posed as an open question by Dillencourt, Eppstein, and Hirschberg~\cite{DillencourtEH00}; see also the selected open questions in the proceedings of GD 2003~\cite{brandenburg2004selected}. 
While this sets both problems on the same level in terms of lower bounds,
there still remains a huge gap between the two notions in the complementary setting of identifying the boundaries of tractability for computing the thickness of a graph. 

Indeed, in 2007 Dujmovi\'{c} and Wood~\cite{DujmovicW07} showed that both notions of thickness are upper-bounded by $\lceil \frac{k}{2}\rceil$ in every graph of \emph{treewidth} $k$. For graph-theoretic thickness, this result in combination with the well-established machinery of Courcelle's Theorem~\cite{Courcelle90} immediately implies that graph-theoretic thickness can be computed in linear time on all graphs of bounded treewidth; stated in terms of the more modern \emph{parameterized complexity} paradigm, the problem is \emph{fixed-parameter tractable} when parameterized by treewidth.
While this provides a robust and classical graph parameter that can be exploited to compute graph-theoretic thickness, the same approach cannot be replicated when aiming for geometric thickness. In fact, no parameterized algorithms were previously known for computing the geometric thickness at all. The overarching aim of this article is to change this via a detailed investigation of the parameterized complexity of computing geometric thickness. 

\subsection{Contributions}
There are two classical perspectives through which one typically applies parameterized analysis to identify tractable fragments for a problem of interest: parameterizing by the solution quality (e.g., the size of a sought-after set), or by structural properties of the input (such as various graph parameters). The \NP-hardness of verifying whether a graph has geometric thickness $2$~\cite{DurocherGM16} effectively rules out the former approach, and so we begin our investigation by focusing on the second perspective. 

\emph{Treewidth}~\cite{RobertsonS86} typically represents a natural first choice for a structural graph parameter one could use to overcome the general intractability of a problem of interest. Indeed, there are numerous examples of combinatorial problems that are known to be fixed-parameter tractable parameterized by treewidth, and the same holds for the aforementioned problem of computing graph-theoretic thickness. Yet, when dealing with graph problems that involve geometric aspects---such as the computation of the \emph{obstacle number}~\cite{BalkoCG00V022}
or \emph{Right-Angle Crossing (RAC) drawings}~\cite{BrandGRS23}---one often finds that treewidth and other ``decomposition-based'' graph parameters (including, e.g., \emph{pathwidth}~\cite{RobertsonS83} and \emph{treedepth}~\cite{sparsity}) are incredibly challenging to use. We refer readers to the survey of Zehavi~\cite{Zehavi22} for further examples of the many open questions surrounding the algorithmic application of such parameters in the context of graph drawing.

Given the above, for our first result we instead consider the computation of the geometric thickness when parameterized by the \emph{vertex cover number}---i.e., the size of a minimum vertex cover of the input graph. In particular, we show:
\begin{restatable}{theorem}{fptvcn}
\label{thm:fptvcn}
\GTfull is fixed-parameter tractable when parameterized by the vertex cover number of the input graph.
\end{restatable}

We prove \cref{thm:fptvcn} through the kernelization technique, which iteratively reduces the size of the input instance while preserving its geometric thickness. Kernelization may be considered the ``staple'' approach for vertex cover number, but its application here is non-trivial and specific to the problem at hand: it relies on a sequence of arguments establishing that every bounded-thickness drawing of a sufficiently large instance must contain a certain configuration of ``similar'' vertices drawn in the plane. We use this to prove that the geometric thickness will not decrease if we remove a specific vertex from the graph.

For our second result, we turn towards the \emph{feedback edge number}---i.e., the edge deletion distance to acyclicity---as a second candidate parameter. The feedback edge and vertex cover numbers can be seen as complementary to each other: the two parameters are pairwise incomparable and form basic restrictions on the structure of the input graph. We show:
\begin{restatable}{theorem}{fptfen}
\label{thm:fptfen}
\GTfull is fixed-parameter tractable when parameterized by the feedback edge number of the input graph.
\end{restatable}

The proof of \cref{thm:fptfen} also relies on kernelization, but the approach is entirely different from the one used for \cref{thm:fptvcn}. In particular, after some simple preprocessing steps, we show that either the instance already has bounded size or one can identify a path $P$ of degree-$2$ vertices such that deleting $P$ preserves the geometric thickness of the input graph. 

While the vertex cover and feedback edge numbers have found successful applications in graph drawing~\cite{BannisterCE18,BhoreGMN20,BalkoCG00V022,BrandGRS23,BinucciGLLMNS24} as well as numerous other settings~\cite{FellowsLMRS08,FominLMT18,BalabanGR24}, they are still highly ``restrictive'' in the sense that they only attain low values on graphs which exhibit rather simple structural properties. 
In light of the aforementioned obstacles standing in the way of developing efficient algorithms that rely on less restrictive structural parameters such as treewidth, the second half of our article takes a different approach. There, we investigate a third perspective to identify meaningful tractable fragments for computing geometric thickness: instead of targeting well-structured graphs, we consider instances where a large part of the solution is already pre-determined (i.e., provided by a user or as the output of some preceding process).
This is formally captured through the setting of \emph{solution extension} problems, which was pioneered in the seminal paper on extending planar drawings~\cite{AngeliniBFJKPR15} and has since then led to a new perspective for studying the complexity of classical graph drawing problems; in this setting, fixed-parameter algorithms were obtained for extending, e.g., \emph{1-planar}~\cite{EibenGHKN20,EibenGHKN20b}, \emph{crossing-optimal}~\cite{GanianHKPV21}, \emph{planar orthogonal}~\cite{BhoreGKMN23} and \emph{linear} \cite{depian2024parameterized} drawings.

For \GTEfull (\GTE\ in short), we consider the input to also contain an edge-partitioned drawing of a subgraph as the partial solution.
As our first (and also technically simplest) result in the extension setting, we provide a linear-time algorithm that can add a constant number of missing edges while preserving a bound on the geometric thickness of the drawing:

\begin{restatable}{theorem}{fptextend}
\label{thm:fptextend}
\GTE when only $k$ edges are missing from the provided partial drawing is fixed-parameter tractable when parameterized by $k$.
\end{restatable}

A major restriction in the above setting is that the partial drawing must already contain all of the vertices of the graph---i.e., one is not allowed to add new vertices. Indeed, typically one allows the addition of vertices and edges in a drawing extension problem, and a common parameter is simply the number of elements missing from the drawing. Surprisingly, we show that a result analogous to the previously mentioned \FPT-algorithms cannot be obtained for geometric thickness: while the problem is polynomial-time solvable when the task is to add a constant number of missing edges and vertices (which we show by providing and analyzing a formulation of the problem in the \emph{Existential Theory of Reals}~\cite{Schaefer09}), we rule out fixed-parameter tractability via a non-trivial reduction.

\begin{restatable}{theorem}{xpwextend}
\label{thm:xpwextend}
\GTE when only $k$ edges and vertices are missing from the provided drawing is \XP-tractable and \W{1}-hard when parameterized by $k$.
\end{restatable}

As our final result, we complete our analysis of the extension setting by showing that the problem remains \NP-hard even if the partial drawing is only missing two vertices and their incident edges, i.e., if the vertex deletion distance to the original graph is $2$:

\begin{restatable}{theorem}{npextend}
\label{thm:npextend}
\GTE is \NP-hard even if the drawn subgraph can be obtained from the input graph by deleting only two vertices.
\end{restatable}

We remark that \cref{thm:npextend} contrasts parameterized algorithms that exploit the vertex deletion distance to solve the drawing extension problem, e.g., for \emph{IC-planar}~\cite{EibenGHKN20} and \emph{1-planar}~\cite{EibenGHKN20b} drawings. 

\ifshort
    \medskip
    \emph{Due to space constraints, we provide full proofs and details of results marked with \appsymb\ in the full version of this paper \cite{this_paper_arxiv_version}.}
\fi

\setbool{inintro}{false}
\subsection{Preliminaries}
\label{sec:prelims}
We assume familiarity with standard graph terminology~\cite{diestel} and
the basic concepts of the parameterized complexity paradigm, specifically \emph{fixed-parameter tractability}, \XP-\emph{tractability}, \W{1}-\emph{hardness} and \emph{kernelization}~\cite{DowneyF13,CyganFKLMPPS15}. For an integer $p$, we set $[p]=\{1,\dots,p\}$. For a function $f \colon A\to B$ and $X\subseteq A$, let $f(X)$ denote the set $\{f(x) \mid x\in X\}$.

Let $G$ be a simple graph with vertex set $V(G)$ and edge set $E(G)$. We use $\overline{E(G)}$ to denote the complement edge set, i.e., $\overline{E(G)}=\binom{V(G)}{2} \setminus E(G)$, and $N_G(v)$ to denote the  neighborhood of a vertex $v$ in $G$ excluding $v$. A \emph{straight-line drawing} $\Gamma$ of $G$ is a mapping from vertices to distinct points in $\mathbb{R}^2$, where we consider edges to be drawn as straight-line segments connecting their end vertices. 

A \emph{geometric $\ell$-layer drawing} of a graph $G$ is a straight-line drawing $\Gamma$ of $G$ accompanied with a function $\chi$ mapping each edge of $G$ to a \emph{planar layer} (also called \emph{color}) from $[\ell]$ with the  property that no pair of edges assigned to the same planar layer cross each other in $\Gamma$. We equivalently say that ($\Gamma, \chi$) is free of \emph{monochromatic crossings}. We can now define \textsc{Geometric Thickness} formally:

\problembox{Geometric Thickness (GT)}{A graph $G$ and an integer $\ell$.}{Does there exist a geometric $\ell$-layer drawing $(\Gamma,\chi)$ of $G$?}

\textsc{GT} on multi-graphs is known to be $\exists\mathbb{R}$-complete~\cite{ForsterKMPTV24}, and thereby it is in particular decidable on simple graphs. 
We will frequently make use of the following folklore observation:
\begin{observation}\label{observation:general_position}
    Let $G$ be a graph.
    If $G$ admits a \drawing{\ell}, it also admits a \drawing{\ell} in general position, i.e., where no three vertices are collinear and no three edges cross at the same point. 
\end{observation}

A geometric $\ell$-layer drawing $(\Gamma_G,\chi_G)$ of $G$ is an \emph{extension} of a geometric $\ell$-layer drawing $(\Gamma_H,\chi_H)$ of a subgraph $H$ of $G$ if $\Gamma_G$ and $\chi_G$ are extensions of $\Gamma_H$ and $\chi_H$, respectively. That is, $(\Gamma_G,\chi_G)$ and $(\Gamma_H,\chi_H)$ assign the same values to the vertices and edges, respectively, of $H$.
With this, we can formalize \GT\ in the extension setting.%

\problembox{Geometric Thickness Extension (GTE)}
{A graph $G$, an integer $\ell$, a geometric $\ell$-layer drawing $(\Gamma_H,\chi_H)$ of some subgraph $H$ of $G$.}
{Does there exist a geometric $\ell$-layer drawing $(\Gamma_G,\chi_G)$ of $G$ which extends $(\Gamma_H,\chi_H)$?}

Observe that \GT\ can be seen as a special case of \GTE\ where $H$ is the empty graph.  We remark that while the problems are formalized as decision problems for complexity-theoretic purposes, every algorithmic result presented within this article is constructive and can be extended via standard techniques to also output a geometric $\ell$-layer drawing as a witness.

\section{Parameterizing by the Vertex Cover Number}
\label{section:gt_by_vc_fpt}

\toappendix{
   \ifshort
   \section{Additional Material for \cref{section:gt_by_vc_fpt}}
   \label{app:section:gt_by_vc_fpt}
   \fi
}

In this section, we show that \probname{Geometric Thickness} parameterized by the vertex cover number is fixed-parameter tractable (\cref{thm:fptvcn}).
On a high level, our proof strategy is as follows. 
We say two vertices are \emph{clones} in a \drawing{\ell} of a graph $G$ if they have the same neighborhood in $G$ and each edge incident to one vertex is colored the same as the respective edge incident to the other vertex.
Given a vertex cover of the input graph and a drawing of it, we partition the plane into a small number of \emph{cells}.
We first formulate a characterization that gives necessary and sufficient conditions for when a vertex in the drawing can be ``cloned'' within the cell it is positioned in without changing the outcome of the instance.

Towards subsequently obtaining a kernel,
we partition the vertices outside the vertex cover into equivalence classes
based on their neighborhood w.r.t.\ the vertex cover.
If an equivalence class is sufficiently large,
we delete an arbitrary vertex $x$ of the class.
Clearly, deleting a vertex preserves positive instances.
Conversely, in a drawing of the instance without $x$, using the pigeon-hole principle,
we will always be able to find at least two mutual clones $c_1, c_2$ of $x$'s class that share a cell. Using our characterization, this will imply that $c_1$ can be cloned within that cell, i.e., $x$ can always be re-inserted into the drawing, positioned inside the cell, and its incident edges can be colored to match those incident to $c_1$.
We thereby obtain a kernel with size dependent on the vertex cover number and the number of colors,  and later show that the latter can be dropped from the parameterization.
Note that, similar to the work of Bannister et al.\ \cite{bannister2013parameterized}, as the graphs of bounded vertex cover number are well-quasi-ordered under induced subgraphs \cite{nevsetvril2012sparsity},
the weaker notion of nonuniform  fixed-parameter tractability \cite{CyganFKLMPPS15} for this problem is already implied, although this neither yields an actual \FPT\ algorithm nor provides concrete runtime bounds.

\paragraph{A Characterization of Cloneability within a Cell.}
To derive the characterization, we start by defining two simple geometric constructions (see \cref{figure:H_and_T}).

\begin{definition}\label{definition:H_and_T}
Given a graph $G$, a drawing $\Gamma$ of $G$ in general position, and three distinct vertices $v, a, b \in V(G)$, we use $H_\Gamma(v, a, b)$ to denote the interior of the half plane defined by the line through $\Gamma(a), \Gamma(b)$ that does not contain $\Gamma(v)$.
Additionally, we use $T_\Gamma(v, a, b)$ to denote the interior of the ``tie''-shape defined by the center point $\Gamma(v)$ and directions $\vv{\Gamma(v)\Gamma(a)}$ and 
$\vv{\Gamma(v)\Gamma(b)}$. %
\end{definition}

\begin{figure}[t]
    \centering
    \includegraphics[page=22]{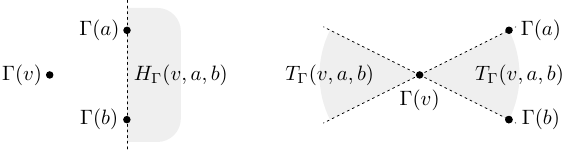}
    
    \vspace{-0.2cm}\caption{Illustration of \cref{definition:H_and_T}.}
    \label{figure:H_and_T}
\end{figure}

\looseness=-1
The characterization is based upon two observations regarding the possible positions where a vertex (or a clone thereof) can be placed while avoiding monochromatic crossings,
each time considering different sets of edges. The two observations follow directly from the definitions; see also Figure~\ref{figure:drawing_maps_clones_in_set}.

\begin{observation}\label{obs:drawing_maps_clones_in_set}
    Let $G$ be a graph, $(\Gamma, \chi)$ be a  \drawing{\ell} of $G$ in general position, and $a,b \in N_G(v)$ be distinct such that $\chi(va) = \chi(vb)$.
    Consider the tuple $(\Gamma', \chi')$ obtained from $(\Gamma, \chi)$ by cloning $v$ into a new vertex $w$ placed in general position. Then, there is no monochromatic crossing between two edges $e_1, e_2$ with $e_1 \in \Set{va, vb}$ and $e_2 \in \Set{wa, wb}$ in $(\Gamma', \chi')$ if and only if $\Gamma'(w) \in T_\Gamma(v, a, b) \cup H_\Gamma(v, a, b)$. 
\end{observation}

The next observation implies there is a small disk of ``safe placements'' around each vertex $v$ in a drawing, where a clone of $v$ can be safely inserted, ignoring crossings with $v$'s edges:

\begin{observation}\label{lemma:global_crossing_free_region}
Let $G$ be a graph, $(\Gamma, \chi)$ be a  \drawing{\ell} of $G$ in general position, and $v \in V(G)$.
Then, the set $B \subseteq \mathbb{R}^2$ where $v$ can be moved to in $(\Gamma, \chi)$ without introducing monochromatic crossings, is open.
\end{observation}

\begin{figure}[t]
    \centering
    \includegraphics[page=23]{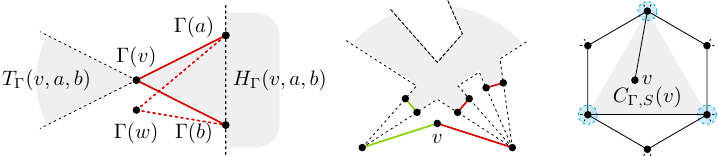}
    
    \vspace{-0.2cm}\caption{Illustrations for \obsref{obs:drawing_maps_clones_in_set} (left), \obsref{lemma:global_crossing_free_region} (middle, the gray area is the complement of $B$), and \cref{definition:cells_and_admissible_region_and_clones} (right, 
    the set $S$ is shown in blue, 
    the cells of $\Gamma$ induced by $S$ are shown in gray).}
    
    \label{figure:drawing_maps_clones_in_set}
    \label{figure:global_crossing_free_region}
    \label{figure:cells}
\end{figure}

With the above setup in hand, we can now turn to the notion of cells and admissible regions.
The intuition behind the admissible region of a vertex is that we consider the restrictions on the position of a clone imposed by all pairs of same-colored edges incident to the vertex (for a single pair, this is precisely \obsref{obs:drawing_maps_clones_in_set}) after disregarding the ``half plane portion'', since we are only interested in potential positions ``close'' to the vertex (i.e., in the same cell).
An illustration for the notion of a cell is  provided in \cref{figure:cells}.

\begin{definition}\label{definition:cells_and_admissible_region_and_clones}
    Let $G$ be a graph, $S$ a vertex cover of $G$, and $\Gamma$ a straight-line drawing of $G$ in general position. For a vertex $v \in V(G) \setminus S$,
    we define the \emph{cell of $v$ in $\Gamma$ with respect to $S$} as the complement of the union of half-planes $H_\Gamma(v,\cdot,\cdot)$ induced by $\Gamma(S)$, i.e.: 
        \begin{equation*}
            C_{\Gamma, S}(v) \coloneqq \mathbb{R}^2 \setminus \bigcup_{a,b \in S, a \neq b} H_\Gamma(v, a, b).
    \end{equation*}
    Furthermore, we call the set obtained by intersecting tie-shapes defined by same-colored edges
    \begin{equation*}
        A_{\Gamma, \chi}(v) \coloneqq \bigcap_{\substack{a, b \in N_G(v), a \neq b, \\ \chi(va) = \chi(vb)}} T_\Gamma(v, a, b)
    \end{equation*}
    the \emph{admissible region} of $v$ in $(\Gamma,\chi)$.
\end{definition}

We can now provide a characterization of when vertices may be cloned inside their cell:

\begin{restatable}{lemma}{localclonablecharacterisation}\label{lemma:local_clonable_characterisation}
    Let $G$ be a graph, $(\Gamma, \chi)$ be a \drawing{\ell} of $G$ in general position, $S$ a vertex cover of $G$, and $v \in V(G) \setminus S$.
    Then, cloning $v$ and placing its clone into the cell $C_{\Gamma, S}(v)$ can yield a \drawing{\ell} if and only if $A_{\Gamma,\chi}(v) \neq \emptyset$.
\end{restatable}

\ifshort
\notappendixproof{lemma:local_clonable_characterisation}{localclonablecharacterisation}{
\begin{proof}[Sketch]
    $(\Rightarrow) \colon$
    Let $(\Gamma', \chi')$ be a \drawing{\ell} in general position obtained from $(\Gamma, \chi)$ by cloning $v$ into the cell $C_{\Gamma, S}(v)$; call the clone $w$.
    Applying \obsref{obs:drawing_maps_clones_in_set} for all distinct $a, b \in N_G(v)$ where $\chi(va) = \chi(vb)$ yields
    \begin{equation*}
        \Gamma'(w) \in \bigcap_{\substack{a, b \in N_G(v), a \neq b, \\ \chi(va) = \chi(vb)}} T_{\Gamma}(v, a, b) \cup H_{\Gamma}(v, a, b).
    \end{equation*}
    Observe that $H_{\Gamma}(v, a, b) \cap C_{\Gamma, S}(v) = \emptyset$ for all distinct $a, b \in N_G(v)$. Hence, using $\Gamma'(w) \in C_{\Gamma, S}(v)$, we obtain
    \begin{equation*}
        \Gamma'(w) \in \bigcap_{\substack{a,b \in N_G(v), a \neq b, \\ \chi(va) = \chi(vb)}} T_\Gamma(v, a, b) = A_{\Gamma, \chi}(v).
    \end{equation*}
    Therefore, $A_{\Gamma, \chi}(v) \neq \emptyset$.

    $(\Leftarrow) \colon$
    We find a placement for the clone inside the cell of $v$ that is ``very close'' to $v$ and in $A_{\Gamma, \chi}(v)$. Correctness then follows from 
    \obsref{obs:drawing_maps_clones_in_set} and \obsref{lemma:global_crossing_free_region}. For the full proof, see the full version of this paper \cite{this_paper_arxiv_version}.
\qed\end{proof}
} %
\fi

\appendixproof{lemma:local_clonable_characterisation}{localclonablecharacterisation}{
\begin{proof}
    $(\Rightarrow) \colon$
    Let $(\Gamma', \chi')$ be a \drawing{\ell} in general position obtained from $(\Gamma, \chi)$ by cloning $v$ into the cell $C_{\Gamma, S}(v)$; call the clone $w$.
    Applying \obsref{obs:drawing_maps_clones_in_set} for all distinct $a, b \in N_G(v)$ where $\chi(va) = \chi(vb)$ yields
    \begin{equation*}
        \Gamma'(w) \in \bigcap_{\substack{a, b \in N_G(v), a \neq b, \\ \chi(va) = \chi(vb)}} T_{\Gamma}(v, a, b) \cup H_{\Gamma}(v, a, b).
    \end{equation*}
    Observe that $H_{\Gamma}(v, a, b) \cap C_{\Gamma, S}(v) = \emptyset$ for all distinct $a, b \in N_G(v)$. Hence, using $\Gamma'(w) \in C_{\Gamma, S}(v)$,
    \begin{equation*}
        \Gamma'(w) \in \bigcap_{\substack{a, b \in N_G(v), a \neq b, \\ \chi(va) = \chi(vb)}} T_\Gamma(v, a, b) = A_{\Gamma, \chi}(v).
    \end{equation*}
    Therefore, $A_{\Gamma, \chi}(v) \neq \emptyset$.

    $(\Leftarrow) \colon$
    Assume $A_{\Gamma,\chi}(v) \neq \emptyset$.
    We find a suitable position to insert a clone of $v$ as follows:
First, we find an open disk $D_1 \subseteq C_{\Gamma, S}(v)$ with center $\Gamma(v)$ and positive radius.
Next, we find the region $B$ (see \obsref{lemma:global_crossing_free_region}) where $v$ can be moved without introducing monochromatic crossings, and use $\Gamma(v) \in B$ together with the fact that $B$ is open (\obsref{lemma:global_crossing_free_region}) to find an open disk $D_2 \subseteq B$ with positive radius and center $\Gamma(v) \in B$.
The set $D \coloneqq D_1 \cap D_2$ is again an open disk with positive radius and center $\Gamma(v)$.
Observe that $A_{\Gamma, \chi}(v)$ forms a ``tie''-shape with center $\Gamma(v)$ as it is a non-empty intersection of ``tie''-shapes of center $\Gamma(v)$.
Hence, there is $x \in D \cap A_{\Gamma, \chi}(v)$ and clearly we can also assume that $\Set{x} \cup \Gamma(V(G))$ is again in general position.

To construct the required drawing, we clone $v$ and insert the clone, call it $w$, at position~$x$.
Because $x \in D_1$, $w$ is inserted into the required cell $C_{\Gamma, S}(v)$.
Because $x \in D_2$ and \obsref{lemma:global_crossing_free_region}, the edges incident to $w$ do not create monochromatic crossings with edges not incident to $v$ in the original drawing.
An edge that is the only one of its color incident to $w$ cannot create a crossing with any edge incident to $v$.
For colors that have two (or more) edges incident to $w$ we cannot get monochromatic crossings with the edges incident to $v$ due to $x \in A_{\Gamma, \chi}(v)$ and \obsref{obs:drawing_maps_clones_in_set}.
Clearly, the edges incident to $w$ in the new drawing also do not create monochromatic crossings with themselves.
Hence, in total, the new drawing satisfies the conditions of this lemma.
\qed\end{proof}
} %

\paragraph{Deriving the Kernel.}

Using Theorem 28.1.1 of the Handbook of Discrete and Computational Geometry \cite{GeometryHandbook} to bound the number of distinct cells in a drawing, we have all prerequisites to derive a kernel parameterized by the vertex cover number plus the number of colors, $\ell$:

\begin{restatable}{lemma}{geomthicknessfptbyvcplusl}\label{lemma:geom_thickness_fpt_by_vc_plus_l}
\probname{Geometric Thickness} parameterized by $\ell + k$, where $k$ denotes the vertex cover number of the input graph and $\ell$ the number of allowed colors in the drawing, admits a problem kernelization mapping each instance $I = (G, \ell)$ to an equivalent instance $I' = (G^*, \ell)$ where $\Size{V(G^*)} = \ell^{\BigO{k}}$.
\end{restatable}
\ifshort

\fi

\appendixproof{lemma:geom_thickness_fpt_by_vc_plus_l}{geomthicknessfptbyvcplusl}{
\begin{proof}
Let $I = (G, \ell)$ be an instance of \probname{Geometric Thickness} and let $S \subseteq V(G)$ be a vertex cover of $G$ of size $k'\leq 2k$ computed, e.g., via the standard 2-approximation algorithm. 
For two vertices $u$, $v$ in $V(G)\setminus S$, we define the equivalence relation $\sim$ where $u \sim v$ if and only if $N_G(u) = N_G(v)$.
The kernel $G^*$ is obtained by applying the following data reduction rule exhaustively:

If there is a vertex $v \in V(G) \setminus S$
whose equivalence class w.r.t.\ $\sim$ (written as $[v]_\sim$) exceeds a size of  
$\ell^{k'} ( \frac{{k'}^2+k'+2}{2} + 1 )$ vertices,
delete $v$ from $G$.

Clearly, the kernel can be computed in polynomial time and \textsc{GT} is decidable since it is in $\exists\mathbb{R}$~\cite{ForsterKMPTV24}. As there are at most $2^{k'}$ equivalence classes w.r.t.\ $\sim$ in $G$ and in the reduced instance, each class has at most $\BigO{\ell^{k'} {k'}^2}$ vertices, the kernel has $\BigO{ 2^{k'} \ell^{k'} {k'}^2} = \ell^{\BigO{k}}$ vertices in total.

It remains to show that the kernelization is \emph{correct}, that is, $I$ is a positive instance of \probname{Geometric Thickness} if and only if $I'$ is a positive instance of \probname{Geometric Thickness}.
For this it suffices to show that the reduction rule is \emph{safe}, that is,
an application of the rule does not change the outcome of an instance.
Let $x$ be a vertex deleted by the rule acting on the instance $(G, \ell)$, producing the instance $(G', \ell)$.
We show that $(G, \ell)$ is a positive instance if and only if $(G', \ell)$ is one:

$(\Rightarrow) \colon$ Clearly, in any \drawing{\ell}, deleting a vertex does not introduce new crossings, hence, the reduced instance is positive as well.

$(\Leftarrow) \colon$
Let $(\Gamma', \chi')$ be a \drawing{\ell} of $G'$ in general position.
Observe that $[x]_\sim \cap V(G')$ has size at least $\ell^{k'} ( \frac{{k'}^2+k'+2}{2} + 1 )$, and that for each vertex in $[x]_\sim \cap V(G')$, the vertex has at most $k'$ incident edges in $G'$, hence its edges can be colored in at most $\ell^{k'}$ ways. Thus, by the pigeonhole principle, there is a set $C \subseteq [x]_\sim \cap V(G')$ of pairwise clones in $(\Gamma', \chi')$ of size at least $\frac{{k'}^2+k'+2}{2} + 1$ all equivalent to $x$ with respect to $\sim$.

By Theorem 28.1.1 in \cite{GeometryHandbook}, there are at most $\frac{{k'}^2+k'+2}{2}$ distinct cells induced by $S$ in $\Gamma'$.
Hence, by the pigeonhole principle, there are distinct $c_1, c_2 \in C$ with $C_{\Gamma', S}(c_1) = C_{\Gamma', S}(c_2)$.
Thus, $c_2$ is a witness certifying that $c_1$ can be cloned into their shared cell in the drawing $(\Gamma'', \chi'')$ obtained from $(\Gamma', \chi')$ by deleting $c_2$.
Hence, by \cref{lemma:local_clonable_characterisation}, $A_{\Gamma'', \chi''}(c_1) \neq \emptyset$.
Towards a contradiction, suppose $A_{\Gamma', \chi'}(c_1) = \emptyset$. Then, since $c_1$ and $c_2$ are not neighbors in $G'$, $A_{\Gamma'', \chi''}(c_1) = \emptyset$, a contradiction.
Therefore, $A_{\Gamma', \chi'}(c_1) \neq \emptyset$ and we can apply \cref{lemma:local_clonable_characterisation} to $(\Gamma', \chi')$ and $c_1$ and obtain the desired \drawing{\ell} of $G$.
\qed\end{proof}
} %

\cref{thm:fptvcn} now follows directly from \cref{lemma:geom_thickness_fpt_by_vc_plus_l} in combination with the previously established fact that the geometric thickness $\ell$ is upper-bounded by $\lceil \frac{k}{2}\rceil$ in the class of graphs of treewidth at most $k$ (which also contains all graphs of vertex cover number $k$)~\cite{DujmovicW07}.

\fptvcn*

\iflong\vspace{1cm}\fi%

\section{Parameterizing by the Feedback Edge
Set Number}
\label{section:gt_by_fen_fpt}

\toappendix{
    \ifshort
   \section{Additional Material for \cref{section:gt_by_fen_fpt}}
   \label{app:section:gt_by_fen_fpt}
   \fi
}

\looseness=-1
In this section, we show that \textsc{GT} is also fixed-parameter tractable
when parameterized by the feedback edge number (\cref{thm:fptfen}).
Our proof strategy is as follows:
After some trivial preprocessing steps and handling of special cases,
we obtain a feedback edge set $F$ of a preprocessed input graph $G'$, where $G'$ can be decomposed into a ``small'' part $G_0$ essentially consisting of $F$, and a ``small'' set of potentially ``long'' paths connecting pairs of vertices from $G_0$.
We sort the paths by length in ascending order and add them to $G_0$ sequentially, obtaining $G_1, G_2$ and so forth.
If the $i$'th path is not significantly longer than the number of edges in $G_{i-1}$, we add the path to $G_{i-1}$ and the resulting graph $G_i$ is still ``small''.
Otherwise, we stop and set $G_j \coloneqq G_{i-1}$ as our kernel (while in the case where we never encounter a ``long'' path, $G'$ itself is a kernel). The centerpiece of the correctness argument is a proof that paths that are significantly longer than the previous ones can be removed from the graph without impacting its geometric thickness.

\newcommand{\pathcount}[0]{x}

\iflong
    \paragraph*{Formalizing the Kernel.}
    \label{section:fen_kernel_formally}
\fi

    To formalize this notion, let $(G, \ell)$ be an instance of \textsc{Geometric Thickness} and let $F$ be a minimum feedback edge set of $G$ of size $k$, which can be computed in polynomial time by computing a spanning tree of $G$.
    First, we describe how to compute the kernel. Later, we will argue about its size, correctness, and running time.
    For the following, we assume $\ell \geq 2$, since for $\ell\leq 1$ the problem is polynomial-time solvable.

\begin{figure}[t]
    \centering
    \includegraphics[page=20]{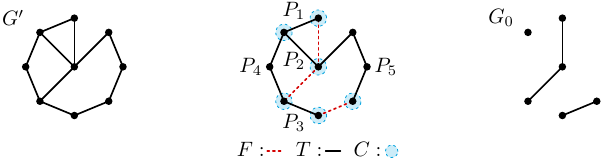}
    \vspace{-0.2cm}\caption{A graph $G'$ without degree one and zero vertices (left), a corresponding feedback edge set $F$ and forest $T$ as well as the vertex set $C$ (middle), and the graph $G_0$ (right).  }
    \label{figure:fen_kernel}
\end{figure}

    Let $G'$ be obtained by exhaustively removing degree-0 and degree-1 vertices from $G$. Consider the forest $T$ with $V(T) \coloneqq V(G')$ and $E(T) \coloneqq E(G') \setminus F$. 
    We define $C$ to be the union of the set $\Set{v \in V(T) \mid \deg_T(v) \ge 3}$ and the set of the endpoints of edges in $F$.
    Next, we consider the decomposition of $T$ into $\pathcount$ edge-disjoint paths whose endpoints are in~$C$; see \cref{figure:fen_kernel}. Observe that the decomposition is unique as $V(T)\setminus C$ only contains degree-2 vertices.
    We refer to these paths, sorted by length in ascending order, by $P_1, \ldots, P_\pathcount$.
    Let $G_0$ be the graph with $V(G_0) \coloneqq C$ and $E(G_0) \coloneqq F$ and
    define $G_i$ as $G_0 \cup P_1 \cup \ldots P_i$ for $i \in [\pathcount]$.
    Observe that $G_\pathcount = G'$. 
    Finally, the kernel is given by $(G_j, \ell)$, where $j$ is the smallest element of $\Set{0, \ldots, \pathcount - 1}$ such that $ |E(P_{j+1})| > 2 \left( |E(G_{j})| + \pathcount \right)$ if such an element exists.
    Otherwise, we set $j \coloneqq \pathcount$.

Clearly, the kernel $G_j$ can be computed in polynomial time and is decidable since \textsc{GT} is in $\exists\mathbb{R}$~\cite{ForsterKMPTV24}. 
Furthermore, using the recursive bounds on the graphs $G_0, \ldots, G_j$, one can show that the kernel has at most $\fenkernelsizeexpr$ vertices\iflong~(\cref{lemma:fen_kernel_size})\fi, and that the instance $(G, \ell)$ is positive if and only if $(G_j, \ell)$ is\iflong~(\cref{lemma:fen_kernel_correctness})\fi. 
For the latter, the key to derive the non-trivial direction of the correctness proof, i.e., how to reintroduce the removed paths without increasing the thickness, is to draw the removed paths as subdivided straight lines.
Since these paths are ``sufficiently long'', all potential monochromatic crossings can be resolved by subdividing and coloring the segments appropriately. 
In total, we have derived \cref{thm:fptfen}.

\fptfen*

\toappendix{
\begin{restatable}{lemma}{fenkernelsize}\label{lemma:fen_kernel_size}
    The graph $G_j$ has at most $\fenkernelsizeexpr$ vertices.
\end{restatable}
\begin{proof}
    First, we show that the number of leaves in $T$ is at most $2k$.
    Let $v$ be a leaf in $T$.
    Since $G'$ does not contain pendant and isolated vertices, $v$ is adjacent to at least two different edges $e_1, e_2$ in $G'$.
    Not both edges can be part of $T$, for then $v$ would not be a leaf in $T$. 
    Thus, either $e_1$ or $e_2$ is in $F$.
    Hence, in total, the number of leaves in $T$ is at most $2|F| = 2k$.
    
    Next, we observe that $|C| \leq 4k$.
    Recall that $C$ is the union of the set of endpoints of edges in $F$ and all vertices in $T$ that have degree at least three in $T$.
    Clearly, the number of underlying vertices of $F$ is at most $2|F| = 2k$, and the number of vertices of degree at least three in $T$ is also upper-bounded by $2k$, i.e., the number of leaves in $T$. Note that the above bound on $|C|$ immediately implies that $\pathcount \leq 4k$.

    Finally, we argue that the computed kernel has at most $\fenkernelsizeexpr$ vertices.
    Let $i \in \{ 0, \ldots, j - 1\} $.
    Observe that $ |E(G_{i+1})| - |E(G_{i})| \leq 2 \left( |E(G_{i})| + \pathcount \right)$ by construction.
    Solving the recursion and using
    $|V(G_0)| = |C| \leq 4k$, $j \leq \pathcount \leq 4k$, $|V(G_j)| \leq 2|E(G_j)|$ finally yields $|V(G_j)| \leq \fenkernelsizeexpr$.
\qed\end{proof}
} %

\toappendix{
\begin{restatable}{lemma}{fenkernelcorrectness}\label{lemma:fen_kernel_correctness}
The instance $(G, \ell)$ is a positive instance of \textsc{\textup{GT}} if and only if $(G_j, \ell)$~is.
\end{restatable}
\begin{proof}
Clearly, $(G, \ell)$ is a positive instance of \textsc{\textup{GT}} if and only if $(G', \ell)$ is, where we recall that $G'$ was obtained by removing vertices of degree at most $1$. 
It remains to be shown that $(G', \ell)$ is a positive instance of \textsc{\textup{GT}} if and only if $(G_j, \ell)$ is.

$(\Rightarrow) \colon$
Let $(\Gamma, \chi)$ be a \drawing{\ell} of $G'$.
Then, since $G_j$ is a subgraph of $G'$, the restriction of $(\Gamma, \chi)$ to $G_j$ is a \drawing{\ell} of $G_j$.

$(\Leftarrow) \colon$
Recalling \obsref{observation:general_position}, let $(\Gamma_j, \chi_j)$ be a \drawing{\ell} of $G_j$ in general position.
We aim to iteratively extend $(\Gamma_j, \chi_j)$ with $P_{j+1}$ up to $P_\pathcount$ and obtain the drawings $(\Gamma_{j+1}, \chi_{j+1})$ up to $(\Gamma_\pathcount, \chi_\pathcount)$, while inductively ensuring that $(\Gamma_i,\chi_i)$ is a \drawing{\ell} of $G_i$ that represents $P_i$ as a subdivided straight-line segment for each $i \in \{ j + 1 , \ldots, \pathcount \}$.

\begin{figure}[t]
    \centering
    \includegraphics[page=16]{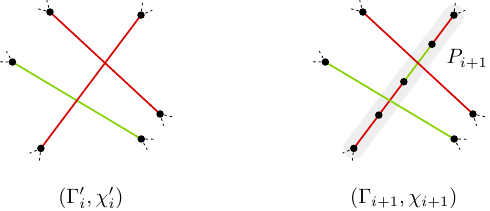}
    \vspace{-0.2cm}\caption{The drawing $(\Gamma'_i, \chi'_i)$ with $P_{i+1}$ represented as a single red edge (left), and $(\Gamma_{i+1}, \chi_{i+1})$ with resolved crossings and $P_{i+1}$ fully drawn (right). In this example, one red crossing needed to be resolved and $P_{i+1}$ contains 3 internal vertices, hence one additional subdivision was undertaken.   }
    \label{figure:resolve_crossing}
\end{figure}

The inductive step is as follows:
Consider $(\Gamma_{i}, \chi_i)$.
We now describe how to obtain $(\Gamma_{i+1}, \chi_{i+1})$ from said drawing by adding $P_{i+1}$. 
Let $(\Gamma_{i}', \chi_{i}')$ be the drawing obtained by 
extending $(\Gamma_{i},\chi_i)$ with a single edge of some fixed color, say red, that connects the starting- and the end-point of $P_{i+1}$.
Now consider the set of monochromatic crossings in $(\Gamma_{i}', \chi_{i}')$:
The single edge representing $P_{i+1}$ may cross with each edge of $G_j$, and furthermore, it may cross with at most one edge of each $P_{j+1}, \ldots P_i$, as these paths are all represented using a subdivided straight-line segment.
Thus, we may obtain at most $|E(G_j)| + \pathcount$ monochromatic crossings.
By construction, $P_{i+1}$
has at least $2 (|E(G_j)| + \pathcount)$ internal vertices.
Hence, we can afford to resolve each monochromatic crossing by subdividing the straight-line segment representation ``closely'' before and after the  crossing, and assigning the new edge connecting the two subdivision vertices a color different than the color of the other edge involved in the monochromatic crossing.
Here, ``closely'' means close enough such that the edge between the two subdivision vertices crosses precisely the other edge involved in the monochromatic crossing.
If after resolving all crossings we have introduced a number of vertices less than the number of internal vertices of $P_{i+1}$, we subdivide the straight-line segment representation sufficiently often at arbitrary positions, inheriting the edge coloring. 
See \cref{figure:resolve_crossing} for an illustration.
\qed\end{proof}
} %

\section{Geometric Thickness Extension}
\label{section:hardness}

\looseness=-1
We analyze the complexity of \textsc{GTE} under three natural parameterizations.
When only edges are missing, using their count as parameter, we obtain an \FPT{} algorithm via a straightforward branching argument\iflong~(see \cref{app:gte_fpt_by_edges})\fi.
In the general setting where also vertices are missing, on one hand we establish \NP-hardness even if all that is missing from the graph are two vertices and their incident edges. This is shown by developing a geometric analogue of the technique of Depian et al.\ \cite{depian2024parameterized}\iflong~(see \cref{section:gte_by_vertices_para_np})\fi, and excludes fixed-parameter as well as \XP-tractability when parameterizing by the vertex deletion distance. On the other hand, in \cref{thm:xpwextend} we show that the problem is in \XP\ and \W{1}-hard parameterized by the total number of missing vertices and edges.

\toappendix{
\ifshort
   \section{Details of the \ETR{}-Formulation}
\fi
}

\looseness=-1
First, we show \XP-membership.
Our strategy is to express \textsc{GT} and \textsc{GTE} as an Existential Theory of the Reals (\ETR{}) formula and to decide said formula using the algorithm by Grigoryev et al.~\cite{GrigoryevV92}. We recall that while membership of \GT\ in $\exists\mathbb{R}$ was recently established \cite{ForsterKMPTV24}, the proof does not provide an explicit formula (which is required for our approach).
Intuitively, an \ETR{} formula is a first order formula where only existential quantification over the reals is allowed,
the terms are real polynomials and $\Set{<, \leq, =, >, \geq}$ are the allowed predicates.
We refer to \cite{toth2017handbook} for a formal definition and introduction to the existential theory of the reals.

\begin{restatable}{lemma}{etrformulation}
\label{lemma:etr_formulation}
    An instance $(G, \ell)$ of \textsc{GT} with $n = |V(G)|$ and $m = |E(G)|$ can be expressed as an \ETR{} formula in $2n + m$ variables with polynomials of total degree at most 6, using \BigO{n^4} polynomial (in)equalities.
\end{restatable}
\appendixproof{lemma:etr_formulation}{etrformulation}{
\begin{proof}
We assume $G$'s vertices and edges to be ordered, that is,
$V(G) = \{ v_1, v_2, \dots, v_{n} \}$ and $E(G) = \{ e_1, e_2, \dots, e_{m} \}$. 
We define an \ETR{} formula $F$ that is satisfiable if and only if $(G, \ell)$ is a positive instance of \textsc{GT} over the variables $\{x_1, x_2, \dots, x_{n}\}$, $\{y_1, y_2, \dots, y_{n}\}$, and $\{c_1, c_2, \dots, c_{m}\}$.
The former two variable sets model the positions of the vertices in the plane, whereas the latter set models the edge-coloring.

Before we give the formula itself, we introduce a shorthand:
Let $i, j, k \in [n]$. We write $A(i, j, k)$ to refer to a full expansion of the determinant
\begin{equation*}
\begin{vmatrix}
    x_i & y_i & 1 \\
    x_j & y_j & 1 \\
    x_k & y_k & 1
\end{vmatrix}.
\end{equation*}
If $\ell > m$, we set $F \coloneqq \top$, as then each edge can simply be assigned a unique color.
Otherwise, we define
\begin{equation*}
    F \coloneqq \exists x_1, x_2, \ldots, x_{n},y_1, y_2, \ldots, y_{n},c_1, c_2, \ldots, c_{m} \colon F_1 \land F_2 \land F_3,
\end{equation*}
where $F_1, F_2, F_3$ are defined as follows:

\iflong\allowdisplaybreaks\fi
\begin{gather*}
    F_1 \coloneqq \bigwedge_{\substack{i,j,k \in [ n ], \\ |\{i,j,k\}|=3}} A(i, j, k) \neq 0,\\
    F_2 \coloneqq \bigwedge_{i \in [m]} \bigvee_{j \in [\ell]} c_i = j \text{, and}\\
    \begin{aligned}
    F_3 \coloneqq \bigwedge_{\substack{e_i = v_av_b, e_j = v_cv_d \in E(G),\\a < b, c < d, |\{a,b,c,d\}|=4}} &c_i = c_j \implies \\
    &( A(a, b, c) A(a, b, d) > 0 \lor A(c, d, a) A(c, d, b) > 0).
    \end{aligned}
\end{gather*}

Here, $F_1$ prescribes that no three vertices be collinear (this we can require due to \obsref{observation:general_position}),
$F_2$ models that the edges be colored using $\ell$ distinct colors,
and $F_3$ prescribes that if two independent edges are assigned the same color, they may not cross.

The line segment intersection test used in $F_3$ is folklore, nonetheless, for completeness sake, we provide a derivation:
Observe that, assuming all $(x_i, y_i)$ are in general position, $A(a, b, c)$ is positive if $(x_c, y_c)$ lies to the left of the line through $(x_a, y_a)$ and $(x_b, y_b)$, and is negative otherwise \cite{toth2017handbook}.
For indices $a, b, c, d$, we write $c \sim_{a b} d$ if $(x_c, y_c)$ and $(x_d, y_d)$ are one the same side of the line through $(x_a, y_a)$ and $(x_b, y_b)$. Clearly, $c \sim_{a b} d$ if and only if $A(a, b, c) A(a, b, d) > 0$.
We now want to find a criterion when the line segments from $(x_a, y_a)$ to $(x_b, y_b)$ and $(x_c, y_c)$ to $(x_d, y_d)$ do not cross. 
Clearly, if $c \sim_{ab} d$ (resp. $a \sim_{cd} b$), the line segments do not cross, as then, one line segment lies completely on one side of the other.
Conversely, it is easy to see that if $c \nsim_{ab} d \land a \nsim_{cd} b$, the line segments do cross.
Hence, the line segments do not cross if and only if $c \sim_{ab} d \lor a \sim_{cd} b$, which is equivalent to $A(a, b, c) A(a, b, d) > 0 \lor A(c, d, a) A(c, d, b) > 0$. 

Finally, observe that the expansion of a single determinant has total degree 3, hence, each polynomial in $F$ has total degree at most 6.
\qed\end{proof}
} %

\toappendix{

\label{section:gte_by_edges_w1_hard}

\newcommand{\bl}[0]{B_l}
\newcommand{\bg}[0]{B_g}
\newcommand{\g}[1]{\widetilde{#1}}

As an immediate corollary, the \XP-tractability part of \cref{thm:xpwextend} follows.
\begin{proof}[\XP-tractability in \cref{thm:xpwextend}]
    Consider an instance $(H, G, (\Gamma_H, \chi_H), \ell )$ of \textsc{GTE}.
    Let $n' \coloneqq |V(G)| - |V(H)|, m' \coloneqq |V(G)| - |V(H)|$, that is, the number of missing vertices and edges, and $k \coloneqq n' + m'$. Furthermore, let $L$ be the number of bits required to encode the reals that define the positions given by $\Gamma_H$ and the colors chosen by $\chi_H$.
    We compute the formula $F$ defined in \cref{lemma:etr_formulation} and substitute all known variables with constants, that is, $(x_i, y_i) \gets \Gamma_H(v_i)$ for all $v_i \in V(H)$ and $c_i \gets \chi_H(e_i) $ for all $e_i \in E(H)$.
    Furthermore, in subformula $F_1$, we delete all conjuncts that contain no free variables, and in subformula $F_3$, we delete all conjuncts that contain no free variables or precisely one free variable, where the other three predetermined positions are collinear. This ensures that the formula remains correct even if $\Gamma_H$ is not in general position.
    The new formula, call it $F'$, has $2n'+m' \leq 2k$ variables.
    Using a suitable {\ETR} solver, e.g.\ \cite{GrigoryevV92}, we can decide $F'$ in time $L^{\BigO{1}} \cdot (9 \cdot \BigO{n^4})^{\BigO{(2n'+m')^2}} = L^{\BigO{1}} \cdot n^{\BigO{k}}$.
    Hence, \textsc{GTE} parameterized by $k$ is \XP-tractable.
    \end{proof}
}

To complete the proof of \cref{thm:xpwextend}, we reduce from the \textsc{Multicolored Clique} problem \cite{CyganFKLMPPS15}.
Note that we obtain \W{1}-hardness even when additionally parameterizing by the total number of layers.
An instance of said problem consists of a $k$-partite graph $X$ with a vertex partition $V(X) = V^1 \cupdot V^2 \cupdot \dots \cupdot V^k$ and an integer $k$ as parameter. The instance is positive if and only if $X$ contains a $k$-clique. 

\begin{figure}[t]
    \centering
    \includegraphics[page=19]{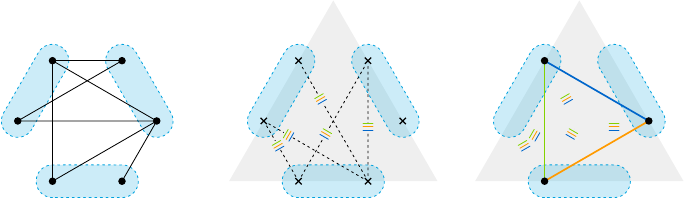}
    \vspace{-0.2cm}\caption{An instance $(X, k=3)$ of \textsc{Multicolored Clique} (left), the resulting instance of \textsc{GTE} where crosses denote possible vertex positions and non-edges of $X$ are drawn as dashed lines (middle), and a valid extension showing that $X$ contains a $K_3$ (right).}
    \label{figure:w1_hardness_idea}
\end{figure}
Our reduction is based on the following idea, illustrated in \cref{figure:w1_hardness_idea}:
Suppose \textsc{GTE} allowed us to specify for each missing vertex a set of possible positions.
We construct an instance of \textsc{GTE} where a drawing is to be extended by a $k$-clique.
For the $i$'th vertex of the $k$-clique, we allow precisely $|V^i|$ possible positions, all placed on the $i$'th side of a regular $k$-gon.
To avoid monochromatic crossings between clique-edges in the extension, we allow for one color for each edge of the $k$-clique.
For each non-edge of $X$, we ``block the visibility'' between the two potential positions corresponding to endpoints of the non-edge via predrawn edges.
Now, if it is possible to draw the clique without monochromatic crossings, the set of chosen positions directly gives the desired $k$-clique in $X$. 
We refer to \iflong~\cref{section:w1_appendix}\else~the full version of this paper \cite{this_paper_arxiv_version}\fi~for the full details of our construction.

\begin{figure}[t]
    \centering
    \includegraphics[page=21]{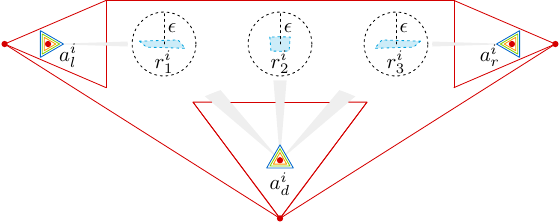}
    \vspace{-0.2cm}\caption{A choice gadget for $V^i$ with $|V^i|=3$ and the clique vertex $c^i$ missing.}
    \label{figure:choice_gadget_simplified}
\end{figure}

 As \textsc{GTE} does not allow us to specify possible vertex positions, our approach is to instead attach $k$ ``choice gadgets'' around the $k$-gon, which allow us to constrain the position of a vertex to a set of disjoint regions. Intuitively, this construction (illustrated in \cref{figure:choice_gadget_simplified}), works as follows: The $i$'th vertex of the clique, call it $c^i$, is connected to three so-called \emph{anchor vertices} $a_l^i, a_d^i, a_r^i$ in the target graph $G$, but not in the predrawn graph $H$, as $c^i \notin V(H)$. Around each anchor vertex, we closely position triangles in all colors except one, say red. This forces the edges from $\{a_l^i, a_d^i, a_r^i\}$ to $c^i$ to be drawn in red in any extension of the given drawing. Around each anchor vertex, we draw a red triangle with specific ``holes'' on the perimeter. For each hole, this induces a cone where $c^i$ can be drawn without creating a red crossing with that triangle. Altogether, the intersection of all cones yields a union of disjoint regions $r_1^i, \dots, r_{|V^i|}^i$, modelling the choice of $v \in V^i$. 
 \iflong
 In \cref{lemma:anchor_segments}, we show that in any valid extension of the drawing, $c^i$ is necessarily drawn inside one of these regions, and in \cref{lemma:regions_in_disks_and_disks_disjoint}, we show that each region is contained inside a disk of radius $\epsilon$ and predetermined center, where we are free to choose $\epsilon$ arbitrarily small.
 \fi
 \ifshort 
One can show that every valid extension requires $c^i$ to be placed inside one of these regions and each region must be fully contained within a disk of sufficiently small radius $\epsilon$ with a predetermined center\iflong~(cf. \cref{section:drawing_single_choice_gadget})\fi. %
 \fi

We arrange the $k$ choice gadgets around a regular $k$-gon. Subsequently, we insert the \emph{blocking edges} that model the structure of $X$. With the positions constrained to disks instead of points, this becomes non-trivial, as the ``lines of sight'' to block obtain a non-zero area. See \cref{figure:w1_reduction_example_short} for an illustration of the lines of sight and the complete reduction\iflong~(\cref{section:w1_appendix})\fi. In particular, we need to be careful not to block ``too much'' (i.e., ensuring edges of $X$ are preserved), while avoiding monochromatic crossings when inserting the blocking edges (i.e., ensuring non-edges of $X$ are preserved).  

Leveraging trigonometry, we calculate a set of permissible positions where we may insert the blocking edges. To ensure that this set is non-empty, we calculate a scaling factor that, when applied to the entire construction, ensures this set to be non-empty. One can show that this construction is always possible\iflong~(cf. \cref{section:drawing_global_blocking_edges})\fi.
The correctness of the reduction follows from the aforementioned properties.
Combined with the \XP-tractability result shown in the previous subsection, we have derived \cref{thm:xpwextend}.

\xpwextend*

\toappendix{
\section{GTE Parameterized by the Number of Missing Edges}
\label{app:gte_fpt_by_edges}

In this subsection, we consider the restricted case of \GTE where only edges are missing from the pre-drawn subgraph $H$ of $G$, i.e., $V(H) = V(G)$.
Let $(\Gamma_H, \chi_H)$ be a geometric $\ell$-layer drawing of $H$.
As all vertex positions are fixed,
it remains to determine the color of the new edges.
Consider an edge $e \in E(G) \setminus E(H)$ that we want to insert into $(\Gamma_H, \chi_H)$.
We can determine in polynomial time the set $C(e) \subseteq [\ell]$ of feasible colors that would not introduce a monochromatic crossing between $e$ and an edge from $H$.
Assume $\Size{C(e)} \geq \Size{E(G) \setminus E(H)}$ and let $(\Gamma_{G \setminus \{e\}}, \chi_{G \setminus \{e\}})$ be a hypothetical geometric $\ell$-layer drawing of the graph $G \setminus \{e\}$ that extends $(\Gamma_H, \chi_H)$. 
Any edge $e' \in E(G) \setminus E(H)$ with $e' \neq e$ that crosses $e$ in $\Gamma_{G \setminus \{e\}}$ removes the color $\chi_{G \setminus \{e\}}(e')$ from $C(e)$.
However, as $\Size{C(e)} \geq \Size{E(G) \setminus E(H)}$, we will still be able to assign a color to $e$ that does not create a  monochromatic crossing, i.e., we can always extend $(\Gamma_{G \setminus \{e\}}, \chi_{G \setminus \{e\}})$ to a geometric $\ell$-layer drawing $(\Gamma_G, \chi_G)$ of the graph $G$ that extends $(\Gamma_H, \chi_H)$.
We formalize this intuition in the following theorem.

\newcounter{oldtheorem}
\setcounter{oldtheorem}{\value{theorem}}
\setcounter{theorem}{\the\numexpr\getrefnumber{thm:fptextend}-1\relax}

\begin{restatable}{theorem}{fptextendstar}
\GTE when only $k$ edges are missing from the provided partial drawing is fixed-parameter tractable when parameterized by $k$.
\end{restatable}

\setcounter{theorem}{\value{oldtheorem}}

\begin{proof}
Let $(H, G, (\Gamma_H, \chi_H), \ell)$ be an instance of \GTE where only $k$ edges are missing from $H$.
As $V(G) = V(H)$, the drawing $\Gamma_G$ of $G$ is fully determined. It remains to compute $\chi_G$ to obtain a geometric $\ell$-layer drawing $(\Gamma_G, \chi_G)$ of $G$ that extends  $(\Gamma_H, \chi_H)$.
Hence, we will in the following focus on computing the coloring $\chi_G$.

We compute for each edge $e = vu \in E(G) \setminus E(H)$ the set $C(e)$ of feasible colors that prevent a monochromatic crossing between $e$ and an edge from $H$.
As the position of $v$ and~$u$ is fixed in $\Gamma_G$, this can be done in polynomial time.
Afterwards, we compute the graph~$\widetilde{G}$ with $V(\widetilde{G}) = V(H)$ and $E(\widetilde{G}) = E(H) \cup \{e \in E(G) \setminus E(H) \mid \Size{C(e)} < k\}$.
Note that $\Size{E(\widetilde{G}) \setminus E(H)} \leq k$ and for every edge $e \in E(\widetilde{G}) \setminus E(H)$, there are less than $k$ feasible colors. %
Hence, in any solution, $\chi_G(e)$ must be one of these colors.
As we want to insert at most $k$ edges, each of which can be colored in at most $k$ colors, we can branch over all possible colorings of the edges of $E(\widetilde{G}) \setminus E(H)$.
This gives us \BigO{k^k} different branches to consider and each branch defines a coloring $\widetilde{\chi_G}$ of the edges of $\widetilde{G}$ that extends $\chi_H$.
Thus, we can obtain in each branch a tuple $(\widetilde{\Gamma_G}, \widetilde{\chi_G})$ and check in polynomial time whether $(\widetilde{\Gamma_G}, \widetilde{\chi_G})$ defines a geometric $\ell$-layer drawing.
Note that if so, $(\widetilde{\Gamma_G}, \widetilde{\chi_G})$ extends $(\Gamma_H, \chi_H)$ by construction. 
Clearly, if in none of the branches the tuple $(\widetilde{\Gamma_G}, \widetilde{\chi_G})$ defines a geometric $\ell$-layer drawing, then there does not exist a geometric $\ell$-layer drawing $(\Gamma_G, \chi_G)$ of $G$ that extends $(\Gamma_H, \chi_H)$.

If a tuple $(\widetilde{\Gamma_G}, \widetilde{\chi_G})$ defines a geometric $\ell$-layer drawing of $\widetilde{G}$, we make the following observation.
Any edge $e \in E(G) \setminus E(\widetilde{G})$ has, by the definition of $\widetilde{G}$, $\Size{C(e)} \geq k$.
Thus, there are at least~$k$ colors available that prevent for $e$ a monochromatic crossing with edges from $H$.
Since $\Size{E(\widetilde{G} \setminus E(H)} < k$, there must exist a color $c \in C(e)$ such that every edge $\widetilde{e} \in E(\widetilde{G})$ that crosses $e$ in $\Gamma_G$ has $\widetilde{\chi_G}(\widetilde{e}) \neq c$.
As we perform the above steps at most $k$ times, each time for a different edge $e \in E(G) \setminus E(\widetilde{G})$ with $\Size{C(e)} \geq k$, we are guaranteed to find for every such~$e$ a color that prevents monochromatic crossings with the already drawn graph.
Hence, we will eventually obtain a geometric $\ell$-layer drawing $(\Gamma_G, \chi_G)$ of $G$ that extends $(\Gamma_H, \chi_H)$, i.e., a solution.
Finally, recall that overall, we have \BigO{k^k} branches to consider, and observe that the steps in each branch run in time linear in the input size.
Therefore, we conclude that the problem is \FPT{} when parameterized by $k$.
\qed\end{proof}

} %

\toappendix{
    
        \section{Details of the \W1-Hardness Proof}
        \label{section:w1_appendix}
    
}

\toappendix{

\begin{figure}[t]
    \centering
    \includegraphics[page=9,width=\linewidth]{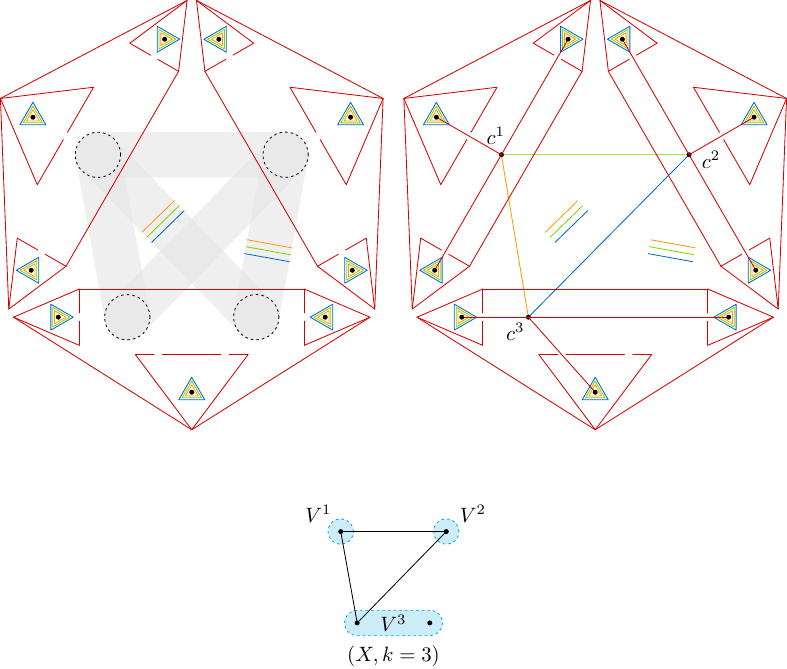}
    \vspace{-0.2cm}\caption{An instance of \textsc{Multicolored Clique} (bottom), the resulting instance of \textsc{GTE} (with the ``tunnels of visibility'' marked in gray), and a witness showing we have a positive instance (right).}
    \label{figure:w1_reduction_example_short}
\end{figure}

\subsection{Formalizing the Reduction}\label{section:formalizing_reduction}
Given an instance of $(X, k)$ of \textsc{Multicolored Clique}, our reduction is as follows:
For each $i \in [k]$, we assume $V^i$ to be ordered as $V^i = \Set{x_1^i, x_2^i, \dots, x_{|V^i|}^i}$.
We also assume that $k \geq 3$, for otherwise, we can trivially decide the instance and map to some trivial positive/negative instance of \textsc{Geometric Thickness Extension}. 

We construct an instance $(H, G, (\Gamma_H, \chi_H), \ell )$ of \textsc{Geometric Thickness Extension} with parameter $k' = 4k + \binom{k}{2}$ and number of colors $\ell = \binom{k}{2} + 1$.
Here, $(\Gamma_H, \chi_H)$ is a \drawing{\ell} of $H$ that needs to be extended to a \drawing{\ell} of $G$. 
The graph $G$ uses the vertex set $V(G) = C \cupdot A \cupdot T_l \cupdot T_d \cupdot T_r \cupdot \bl \cupdot \bg$.
    Here, $C$ denotes the \emph{clique vertices}, $A$ denotes the \emph{anchor} vertices, $T_l, T_d, T_r$ denote the \emph{left/down/right triangle vertices}, $\bl$ denotes the vertex set of the \emph{local blocking edges}, and $\bg$ denotes the vertex set of the \emph{global blocking edges}.
    The clique vertices $C = \Set{c^1, c^2, \dots, c^k}$ form a clique in $G$, and each $c^i \in C$ is connected to three unique anchor vertices $\Set{a_l^i, a_d^i, a_r^i} \subseteq A$.
    We will define the contents of the remaining vertex sets and their associated edges as soon as they are required.
    The subgraph $H$ of $G$ is set to $G$ without the clique vertices, that is, 
    $H \coloneqq G[V(G) \setminus C]$.

        To define $(\Gamma_H, \chi_H)$, we proceed as follows:
        First, we describe how to draw a single choice gadget. The construction will depend on a scaling factor $s$. Then, we describe how to compose $k$ choice gadgets into a single drawing, again with respect to the scaling factor $s$.
        Next, we introduce the so called global blocking edges and find a suitable scaling factor $s_0$ along the way.
        To complete the drawing, we introduce the so called local blocking edges.
        Alongside, we will also define the coloring $\chi_H$, argue why the drawing is free of monochromatic crossings, and collect useful observations for the correctness proof later on.

\begin{figure}[t]
    \centering
    \includegraphics[page=1,width=\linewidth]{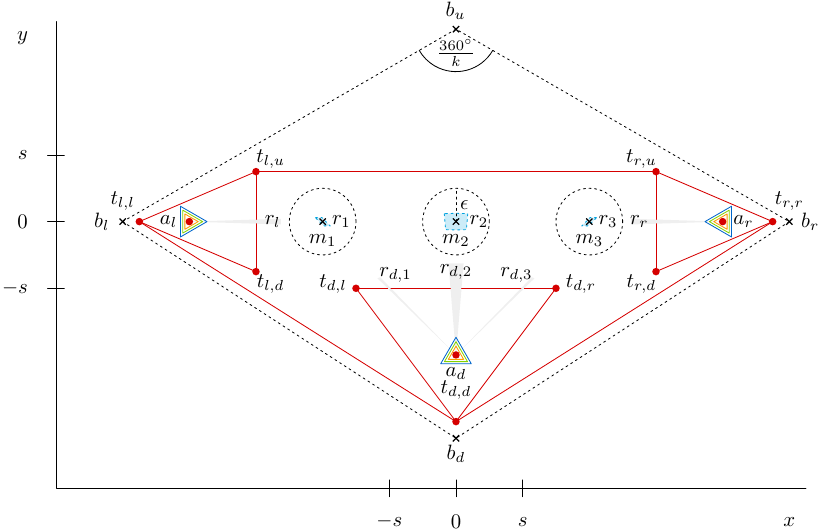}
    \vspace{-0.2cm}\caption{A choice gadget with $k = 3, |V^i| = 3, \epsilon = \frac{1}{2}, s = 1$. The superscript $ { \hspace{0pt} }^i$ is dropped for readability. For illustrative purposes, the local blocking edges of $a_l, a_d, a_r$ are drawn as well. These edges force the edges from $a_l, a_d, a_r$ to the clique vertex to be drawn in red and are discussed in \cref{section:drawing_local_blocking_edges}. See \cref{section:precise_construction_choice_gadgets} for a definition of all objects not described in this subsection.}
    \label{figure:choice_gadget}
\end{figure}

        First, we define a family of drawings $(\Gamma_s)_{s \in \mathbb{R}_{\geq 1}}$.
        Here, $s$ is a scaling factor greater than or equal to one. Later, we will find $s_0$ such that $\Gamma_s$ has desirable properties and set $\Gamma_H \coloneqq \Gamma_{s_0}$.

        Let $s \in \mathbb{R}_{\geq 1}$.
        We now describe how to obtain $\left.\Gamma_s\right|_{V(H) \setminus (\bg \cup \bl)}$, that is, $\Gamma_s$ where the local and global blocking edges are yet to be drawn.

        In what follows, we use the term \emph{red} to refer to color $0$, and define $\epsilon = \min \left\{\frac{3}{2|V^i|} \mid i \in [k]\right\}$.
        Furthermore, we use the notation $\disk{x}{r}$ do denote a closed disk in $\mathbb{R}^2$ with center $x$ and radius $r$, and write $\convexhull(\dots)$ to denote the convex hull of a set of sets or points.
        
        \subsection{Drawing a Single Choice Gadget}    \label{section:drawing_single_choice_gadget}

        We will now describe how to construct the $i$'th choice gadget and derive two  properties a choice gadget guarantees. 

        The purpose of a choice gadget is to constrain the possible position of $c^i$ to one of $|V^i|$ disjoint regions, call them $r^i_1, \dots, r^i_{|V^i|}$, each confined to a disk with radius $\epsilon$ with centers $m^i_1, \dots, m^i_{|V^i|}$ along a line segment. Note that $\epsilon>0$ can in principle be chosen to be arbitrarily small, but the concrete value we fixed above suffices for our purposes.

        Consider \cref{figure:choice_gadget} where a choice gadget is shown.
        A choice gadget is a partially drawn graph.
        This graph contains three \emph{anchor vertices}: $a^i_l$ (left), $a^i_d$ (down), $a^i_r$ (right), which are all adjacent to the clique vertex $c^i$.
        The only vertex not yet drawn is $c^i$.
        Around each anchor vertex, we closely draw $\binom{k}{2}$ triangles in all colors except red (the local blocking edges), which forces that the edges to the clique vertex $c^i$ be drawn in red. Furthermore, around each anchor vertex, we position a red triangle (the left, down, and right triangle) with specific holes in the boundary.
        Each hole induces a cone where $c^i$ can be drawn crossing-free.
        The precise coordinates of all vertices are chosen such that the intersection of all cones gives the set of disjoint regions $r^i_1, \dots, r^i_{|V^i|}$ as described above.

        We describe the precise construction including concrete positions for all vertices in \cref{section:precise_construction_choice_gadgets}.

Proceeding onward, we do not need to keep in mind the precise internal geometry of a choice gadget, but rather it suffices that the construction satisfies the following two properties:
\begin{restatable}{lemma}{regionsindisksanddisksdisjoint}
  \label{lemma:regions_in_disks_and_disks_disjoint}
  Each region $\g{r_j^i}$ is contained in a disk with center $\g{m_j^i}$ and radius $\epsilon$, and all such disks are disjoint.
\end{restatable}

\begin{proof}
Let $i \in [k]$ and $j \in [|V^i|]$.
For ease of notation, we work in the coordinate space of choice gadget $i$ and drop the accent $\g{\;}$.
We observe that since $a_l^i \cap a_r^i$ is bounded vertically by horizontal lines at $y = \pm \sfrac{\epsilon}{4}$, $r_l^i \cap r_r^i \cap r_{d,j}^i = r_j^i$ is bounded by a parallelogram (\cref{figure:regions_in_disks_and_disks_disjoint}, marked with red outline). We consider the axis-aligned bounding box of said parallelogram. The height of the bounding box is $\epsilon/2 < \sfrac{\epsilon}{\sqrt{2}}$. 
Solving for the width of the bounding box to be less than $\sfrac{\epsilon}{\sqrt{2}}$ (using that $\tan(\alpha_j^i) \ge 1$ as per our construction) yields $\zeta^i_j$ be upper-bounded exactly as we have defined $\zeta^i_j$ to be.

Using that $\sfrac{\epsilon}{\sqrt{2}}$ upper-bounds both the width and height of the bounding box, it is easy to see that no matter where $m_j^i$ is located in relation to the bounding box, the box is always contained in the disk $\disk{m_j^i}{\epsilon}$.

The second part of the lemma holds because inside of a single choice gadget, the disks are disjoint by choice of $\epsilon$, and globally, the choice gadgets do not overlap.
\qed\end{proof}

\begin{figure}[t]
    \centering
    \includegraphics[page=4]{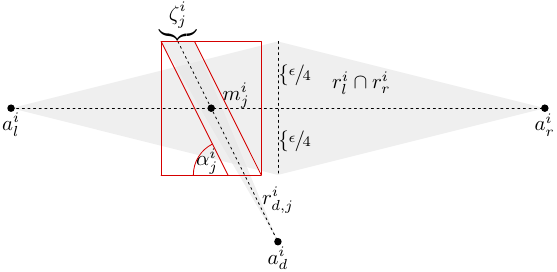}
    \vspace{-0.2cm}\caption{ Illustration for \cref{lemma:regions_in_disks_and_disks_disjoint}.   }
    \label{figure:regions_in_disks_and_disks_disjoint}
\end{figure}

\begin{restatable}{lemma}{anchorsegments}
  \label{lemma:anchor_segments}
  Let $i \in [k]$. If $(\Gamma_H, \chi_H)$ is extended with the clique vertex $c^i$ such that its edges to the anchor vertices $a_l^i, a_d^i, a_r^i$ are drawn in red, then $c^i$ is drawn in the region $\g{r_1^i} \cup \dots \cup \g{r_{|V^i|}^i}$.
\end{restatable}
\begin{proof}
Let $l_l, l_d, l_r$ be line segments starting at a common point that end at the anchor points $\g{a_l^i}$, $\g{a_d^i}$, $\g{a_r^i}$ respectively.
It suffices to show that then all three segments do not cross a red segment in $(\Gamma_H, \chi_H)$
if and only if the common starting point lies in the region $\g{r_1^i} \cup \dots \cup \g{r_{|V^i|}^i}$.

Consider again \cref{figure:choice_gadget}.
It is clear that a red segment is not crossed if and only if the common endpoint lies either inside the union of the left triangle and $\g{r_l^i}$,
and in the union of the right triangle and $\g{r_r^i}$,
and in the union of the downward triangle and $\g{r_{d,1}^i} \cup \dots \cup \g{r_{d,|V^i|}^i}$.
Since the left, right, and downward triangle are disjoint, the statement follows.      
\qed\end{proof}

\subsection{Precise Construction of Choice Gadgets}
\label{section:precise_construction_choice_gadgets}
        We now describe precisely how to construct the $i$'th choice gadget.
        First, we define the remaining vertices and edges required to do so:
        \begin{itemize}
            \item $T_l = \bigcup_{i \in [k]} T_l^i$, where $T_l^i = \Set{t_{l,l}^i, t_{l,u}^i, t_{l,d}^i, t_{l,b,1}^i, t_{l,b,2}^i}$, the \emph{left triangle vertices}, where $l, r, u, d, b$ denote left, right, up, down, and boundary respectively,
        \item $T_r = \bigcup_{i \in [k]} T_r^i$, where $ T_r^i = \Set{t_{r,l}^i, t_{r,u}^i, t_{r,d}^i, t_{r,b,1}^i, t_{r,b,2}^i}$, the \emph{right triangle vertices}, where $b$ stands for \emph{barrier}, and, 
        \item $T_d = \bigcup_{i \in [k]} T_d^i$, where $T_d^i = \Set{t_{d,l}^i, t_{d,u}^i, t_{d,d}^i} \cup \Set{t_{l,b,2j-1}^i, t_{l,b,2j}^i \mid j \in [|V^i|] }$, the \emph{downward triangle vertices}.
        \end{itemize}

For each $i \in [k]$,
        \begin{itemize}
            \item vertices $t_{l,b,1}^i, t_{l,u}^i, t_{l,l}^i, t_{l,d}^i, t_{l,b,2}^i$ form a path,
        \item vertices $t_{r,b,1}^i, t_{r,u}^i, t_{r,r}^i, t_{r,d}^i, t_{r,b,2}^i$ form a path,
        \item vertices $t_{d, b, 1}^i, t_{d, l}^i, t_{d, d}^i, t_{d, r}^i, t_{d, b, 2|V^i| }^i$ form a path, and finally,
        \item for each $j \in [|V^i| - 1]$, vertices $t_{d, 2j}^i$ and $t_{d, 2j+1}^i$ form an edge.
        \end{itemize}

        Refer to \cref{figure:choice_gadget} for the positions of all of the gadget's vertices which are marked with a red disk. The positions are determined relative to the gadgets center point $(0, 0)$.
        More precisely, we position $t_{l, u}$ at $s \cdot (-3, \sfrac{3}{4}), t_{l, d}$ at $s \cdot (-3, -\sfrac{3}{4}), t_{l, l} $ at $ s \cdot (-4 \sfrac{3}{4}, 0), t_{d, d}$ at $ s \cdot (0,-3), t_{d,l} $ at $ s  \cdot (-1 \sfrac{1}{2}, -1), t_{d,r} $ at $ s  \cdot (1 \sfrac{1}{2}, -1), t_{r, u} $ at $ s \cdot (3, \sfrac{3}{4}), t_{r, d}$ at $ s \cdot (3, -\sfrac{3}{4}), t_{r, r} $ at $ s \cdot (4 \sfrac{3}{4}, 0), b_l $ at $ s \cdot (-5, 0), b_r $ at $ s \cdot (5, 0),$ and $b_d $ at $ s \cdot (0, -3 \sfrac{1}{4})$.
        The position of $b_u$ is determined by
        ensuring the angle $\angle b_l b_u b_r$ is $\frac{360^{\circ}}{k}$ and $b_u$'s $y$-coordinate is positive.

        \begin{figure}[t]
            \centering
            \includegraphics[page=2]{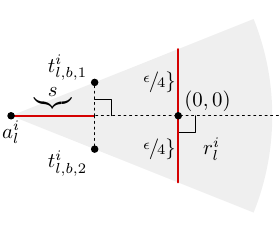}
            \vspace{-0.2cm}\caption{Diagram illustrating how to compute the positions of the barrier vertices of the left triangle.}
            \label{figure:choice_gadget_triangle_left}
        \end{figure}

                \begin{figure}[t]
            \centering
            \includegraphics[page=3]{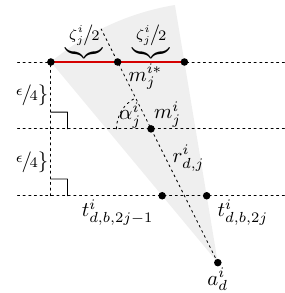}
            \vspace{-0.2cm}\caption{Diagram illustrating how to compute the positions of the barrier vertices of the downward triangle. }
            \label{figure:choice_gadget_triangle_down}
        \end{figure}

        Next, we define a series of points (marked with a black cross) in relation to the gadgets center point $(0, 0)$.
        \begin{itemize}
            \item The \emph{boundary points} $b_l^i, b_d^i, b_r^i, b_u^i$ (left, down, right, up). The positions of $b_l^i, b_d^i, b_r^i$ are defined by \cref{figure:choice_gadget}. The position of $b_u^i$ is obtained by ensuring $ \angle b_l^i b_u^i b_r^i = \sfrac{360^{\circ}}{k}$.
            \item The \emph{middle points} $m_1^i, m_2^i, \dots, m_{|V^i|}^i$ lie on the $x$-axis evenly spaced on the interval $[-2; 2]$. The $x$-coordinate of the $j$'th middle point is
            \begin{equation*}
                s \cdot \left( \frac{4(j-1)}{|V^i|-1}-2 \right)
            \end{equation*}
            if $|V^i| > 1$, otherwise the single middle point is mapped to the origin.            
        \end{itemize}

        Now, we are ready to define the positions of the barrier vertices, that is, the vertices of $T^i_l, T^i_d, T^i_r$ whose positions we have not fixed so far:

        We begin with the left triangle.
        The corresponding barrier vertices are $t_{l,b,1}^i, t_{l,b,2}^i$.
        The position of $t_{l,b,1}^i$ is given by the intersection of the line through $a_l^i$ and $(0, \sfrac{\epsilon}{4})$ and the vertical line through the point $a_l^i + (s, 0)$.
        The point $t_{l,b,2}^i$ is $t_{l,b,1}^i$ mirrored at the $x$-axis.
        See \cref{figure:choice_gadget_triangle_left} for an illustration.

        The barrier vertices of the right triangle are found symmetrically (refer to \cref{figure:choice_gadget}).

        Finally, we deal with the barrier vertices $\Set{t_{l,b,2j-1}^i, t_{l,b,2j}^i \mid j \in [|V^i|] }$ of the downward triangle.
        Let $j \in [|V^i|]$. Let $m_j^{*i}$ be the intersection of the line through $a_d^i$ and $m_j^i$ and the horizontal line through $(0, \sfrac{\epsilon}{4})$. Furthermore, let $\alpha_j^i \in (0, \sfrac{\pi}{2}]$ be the angle at which these two lines cross.
        The number $\zeta^i_j$ is set to a nonzero value less than
        \begin{equation*}
            \frac{-\epsilon + \sqrt{2}\epsilon \tan(\alpha_j^i)}{2 \tan(\alpha_j^i) }.
        \end{equation*}
        Now consider the line through $m_j^{i*} - (\sfrac{\zeta^{i}_j}{2}, 0)$ and $a_d^i$. The intersection of this line with the horizontal line through $(0, -\sfrac{\epsilon}{4})$ gives the position of $t_{l,b,2j-1}^i$.
        Symmetrically, the intersection of this horizontal line and the line through $m_j^{i*} + (\sfrac{\zeta^{i}_j}{2}, 0)$ and $a_d^i$ gives the position of $t_{l,b,2j}^i$.
        In total, this fixes the position of all boundary vertices.
        See \cref{figure:choice_gadget_triangle_down} for an illustration.

        Observe that all boundary vertices are ordered strictly monotone with respect to the line segment pointing to the origin of the respective triangle. Hence, by \cref{figure:choice_gadget}, no edges cross in a single choice gadget.

        Note that for $k > 3$, the $y$-coordinate of $b_u$ is bigger than for $k = 3$.
        Hence we observe, again using \cref{figure:choice_gadget}, that all vertices and edges of the choice gadget are contained in the convex hull of $b_u^i, b_l^i, b_r^i, b_d^i$.
        
        Next, we define a number of regions in the plane in relation to the edge gadgets positions:
        \begin{itemize}
            \item The infinite cone with corner point $a_l^i$ and boundary through $t_{l,b,1}^i$ and $t_{l,b,2}^i$ defines the region $r_l^i$.
            \item Symmetrically, the infinite cone with corner point $a_r^i$ and boundary through $t_{r,b,1}^i$ and $t_{r,b,2}^i$ defines the region $r_r^i$.
            \item For each $j \in [|V^i|]$, the infinite cone with corner point $a_d^i$ and boundary through $t_{d, b,2 j - 1}^i$ and $t_{d, b,2j}^i$ defines the region $r_{d,j}^i$.
            \item Finally, for each $j \in [|V^i|]$, the region $r_j^i$ is given by $r_l^i \cap r_r^i \cap r_{d,j}^i$.
        \end{itemize}
        See \cref{figure:choice_gadget} for an example of all defined regions, \cref{figure:choice_gadget_triangle_left} for the region $r_l^i$, and \cref{figure:choice_gadget_triangle_down} for the regions $r_{d,j}^i$.
        
        \subsection{Composing the Choice Gadgets}
        \label{section:composing_choice_gadgets}

Consider the \emph{boundary points} $b_l^i, b_d^i, b_r^i, b_u^i$ (left, down, right, up) (\cref{figure:choice_gadget}) of a choice gadget.
The choice gadget is contained entirely within the convex hull of its boundary points.

We now describe how to combine all $k$ choice gadgets into a single drawing:
        Let $P$ be a regular $k$-gon and side length $10s$, which equals the distance between the left and right boundary point of each choice gadget.
        For each of the $k$ choice gadgets, pick one of the $k$ sides of $P$.
        We apply a rigid transformation to the $i$'th choice gadget such that the upper boundary $b_u^i$ of the choice gadget is mapped to the center of $P$, the left boundary $b_l^i$ is mapped to one endpoint of the corresponding side of $P$, and that the right boundary $b_r^i$ is mapped to the other endpoint of the corresponding side of $P$.
        
        We transform all vertex positions, points, and regions for each choice gadget this way and denote the result with the accent $\g{\;}$. For example, $\g{b_u^i} = (0, 0)$ for all $i \in [k]$.
        
        Finally, we set $\left.\Gamma_s\right|_{V(H) \setminus (\bg \cup \bl)}$ such that it agrees with the positions we just computed. That is,
        for all $v \in V(G) \setminus \Set{\bg \cup \bl}$, we have $\Gamma_s(v) \coloneqq \g{v}$.

        Note that each choice gadget is crossing-free, and we obtained our composite drawing by combining choice gadget drawings in a non-overlapping way. Hence, $\left.\Gamma_s\right|_{V(H) \setminus (\bg \cup \bl)}$ is crossing free.

        We color all edges drawn so far (i.e. those not incident to a vertex in $\bl \cup \bg$) in red.
        It remains to draw the vertices $\bg$ and $\bl$.

        \subsection{Drawing the Global Blocking Edges}
        \label{section:drawing_global_blocking_edges}

        \begin{figure}[t]
            \centering
            \includegraphics[page=6]{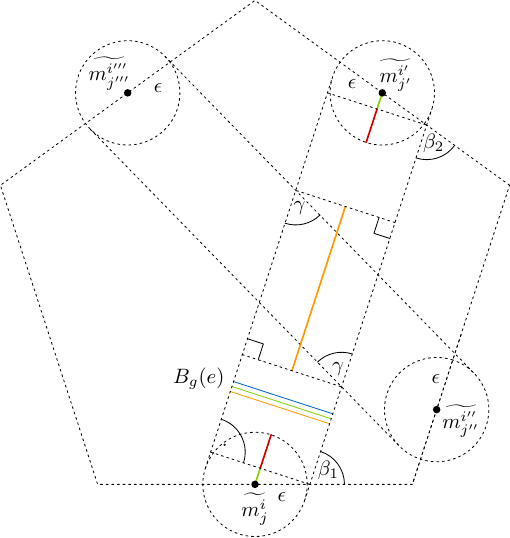}
            \vspace{-0.2cm}\caption{Illustration of the various forbidden line segments and a possible placement of a set of global blocking edges.}
            \label{figure:pentagon}
        \end{figure}

        Next, we describe how to draw the global blocking edges to to obtain $\left.\Gamma_s\right|_{V(H) \setminus \bl}$, that is, $\Gamma_H$ where only the local blocking edges are missing.
        The vertex set $\bg$ of the global blocking edges is composed of
        a set of vertices ${\bg}(e)$ for each non-edge $e$ in $X$, and ${\bg}(e)$ induces a matching of $\binom{k}{2}$ edges (one for each color except red), which we call the \emph{global blocking edges} (of $e$).

        The goal of this section is to draw the global blocking edges such
        that two regions from two choice gadgets that select a vertex of $X$ each can be connected via an edge if and only if the corresponding edge in $X$ exists:
        \begin{lemma}\label{lemma:global_blocking}
            Let $i, i' \in [k]$ be distinct and $j \in [|X_i|], j' \in [|X_{i'}|]$ and consider $\left.\Gamma_s\right|_{V(H) \setminus \bl}$.
            A line segment starting in $\g{r_j^i}$ and ending in $\g{r_{j'}^{i'}}$
            crosses all global blocking edges of $x_j^ix_{j'}^{i'}$ if $x_j^ix_{j'}^{i'} \not\in E(X)$, and crosses no global blocking edges otherwise.
        \end{lemma}       

        Let $e = x_j^i x_{j'}^{i'} \in \overline{E(X)}$ and see \cref{figure:pentagon} for an illustration.
        We draw the global blocking edges of $e$ (with vertex set ${\bg}(e)$)
        in all colors except red\footnote{We exclude red as it is already blocked by the choice gadget construction.} orthogonally through the ``the tunnel of visibility''  between the corresponding regions ($\g{r^{i}_{j}}~$and~$\g{r^{i'}_{j'}}$).

        To do so, we need to find ``free space'' where we can insert the edges safely, i.e. without introducing any crossings with the rest of the drawing, nor global blocking edges of other non-edges of $X$.
        The strategy to do so is as follows:
        We search for a suitable subsegment of the line segment from $\g{m^i_j}$ to $\g{m^{i'}_{j'}}$ by marking all subsegments where we cannot draw the edges safely (of which there are several kinds).
        It may happen that the entire line segment turns out to be unsuitable for placement. To combat this, we carefully compute the scaling factor $s_0$ such that we can always find a safe subsegment.

        See \cref{section:details_global_blocking_edges} for the details of this construction and the proof of 
        \cref{lemma:global_blocking}.

\subsection{Details of Inserting Global Blocking Edges}
\label{section:details_global_blocking_edges}
        To ensure \cref{lemma:global_blocking} holds, it suffices to
        draw the global blocking edges of $e = x_j^i x_{j'}^{i'} \in \overline{E(X)}$ in all colors except red perpendicular to the line segment from $\g{m^i_j}$ to $\g{m^{i'}_{j'}}$ such that the endpoints lie on the boundary of the convex hull formed by the two disks with centers $\g{m^i_j}$ and $\g{m^{i'}_{j'}}$ with radius $\epsilon$,
        that the drawn edges do not intersect with either disk, that all edges are drawn inside $P$, that we do not introduce any crossings (ignoring colors), and that for each edge of $X$, the drawn blocking edges do not intersect the convex hull of the two corresponding disks.
        
        To achieve this, for each non-edge $e$, 
        we compute a series of \emph{forbidden line segments}, which are sub-line segments of the line segment from $\g{m^i_j}$ to $\g{m^{i'}_{j'}}$.
        Drawing all edges of the matching such that they do not cross such a forbidden line segment will ensure that all properties listed above are achieved.

        For each non-edge $e$, we use the following forbidden line segments and use $f_e \in \mathbb{R}$ to keep track of their total length.
        \begin{itemize}
            \item A segment starting at $\g{m^i_j}$ of length $\epsilon$, and symmetrically, a segment starting at $\g{m^{i'}_{j'}}$ of length $\epsilon$.
            These segments ensure that the edges of the matching will not intersect with the two disks. These segments are marked in bold red in \cref{figure:pentagon} (Note that part of the segments is overshadowed by green forbidden segments defined next).
            Their contribution to $f_e$ is $2 \epsilon$.
            \item A segment starting at $\g{m^i_j}$ of length such that the segment-endpoint lies on the boundary of the area-maximal rectangle that can be inscribed into the convex hull of $(\disk{\g{m^i_j}}{\epsilon}$ and $\disk{\g{m^{i'}_{j'}}}{\epsilon}$.
            Symmetrically, a line segment starting at $\g{m^{i'}_{j'}}$.
            These segments serve to enforce that the edges of the matching be drawn inside of $P$ and are marked in bold green in \cref{figure:pentagon}.
            Elementary trigonometry reveals that their contribution to $f_e$ is 
            \begin{equation*}
                \frac{\epsilon}{\tan(\beta_1)} + \frac{\epsilon}{\tan(\beta_2)},
            \end{equation*}
            where $\beta_1, \beta_2 \in (0, \sfrac{\pi}{2}]$ are the angles of the line segment from $\g{m^i_j}$ to $\g{m^{i'}_{j'}}$ and the polygon-side of gadget $i$ and gadget $i'$.
            \item For each line segment from $\g{m^{i''}_{j''}}$ to $\g{m^{i'''}_{j'''}}$ that crosses the line segment (or shares exactly one endpoint) from $\g{m^i_j}$ to $\g{m^{i'}_{j'}}$, we compute a forbidden segment as follows:
            First, compute the parallelogram
            \begin{multline*}
    \convexhull(\disk{\g{m^i_j}}{\epsilon}, \disk{\g{m^{i'}_{j'}}}{\epsilon}) \\
    \cap \\
    \convexhull(\disk{\g{m^{i''}_{j''}}}{\epsilon}, \disk{\g{m^{i'''}_{j'''}}}{\epsilon}).
\end{multline*}
            Then, project the parallelogram orthogonally onto the line trough $\g{m^i_j}$ to $\g{m^{i'}_{j'}}$ to obtain the forbidden segment.
            A forbidden segment of this kind serves to enforce that the global blocking edges corresponding to the current non-edge do not cross with the corresponding global blocking edges of another non-edge, and that the convex hulls corresponding to edges of $X$ remain free of global blocking edges.
            An example of such a segment is marked in bold orange in \cref{figure:pentagon}.
            Elementary trigonometry reveals that their contribution to $f_e$ is
            \begin{equation*}
                \tan( \frac{\pi}{2} - \gamma ) + \frac{1}{\sin(\gamma)},
            \end{equation*}
            where $\gamma \in (0, \pi]$ is the angle at which the two line segments connecting the midpoints $\g{m^i_j}, \g{m^{i'}_{j'}}$ and $\g{m^{i''}_{j''}}, \g{m^{i'''}_{j'''}}$ respectively cross.
        \end{itemize}

        To find a sub-segment where the current global blocking edges ${\bg}(e)$ can be embedded satisfying all conditions listed above, it suffices to subtract all forbidden segments from the line segment from $\g{m^i_j}$ to $\g{m^{i'}_{j'}}$ and to pick an arbitrary sub-segment.
        If the forbidden segments cover the whole segment from $\g{m^i_j}$ to $\g{m^{i'}_{j'}}$, we say $\Gamma_s$ is degenerate and $\Gamma_s$ is undefined.

        Observe that $f_e$ does not depend on $s$.
        Hence, we can select a scaling factor $s_0$ such that each line segment connecting any two midpoints is longer than $f_e$, yielding a non-degenerate $\Gamma_{s_0}$.
        We achieve this by setting
        \begin{equation*}
            s_0 \coloneqq \max_{e \in \overline{E(X)}} \frac{f_e + 1}{l_e}, 
        \end{equation*}
        where for each non-edge $e \in \overline{E(X)}$, $l_e$ is the length of the corresponding line segment connecting the corresponding midpoints.
        
        This construction and \cref{lemma:regions_in_disks_and_disks_disjoint} finally yields \cref{lemma:global_blocking}.

        \subsection{Drawing the Local Blocking Edges}
        \label{section:drawing_local_blocking_edges}

        Finally, we draw the local blocking edges and edges and obtain the completed drawing $(\Gamma_H = \Gamma_{s_0}, \chi_H)$.
        The vertex set $\bl$ of the local blocking edges is composed of
        a set of vertices ${\bl}(a)$ for each anchor vertex $a \in A$, 
        where ${\bl}(a)$ induces $\binom{k}{2}$ disjoint triangles (one for each color except red).
    
        We insert the anchor vertices $A$ one by one in an arbitrary sequence while maintaining the validity of the drawing. Let $a \in A$ be the next anchor vertex in the sequence and
        let $\epsilon_a$
        be small enough such that the disk with center $\Gamma_H(a)$ does not contain any other vertex whose position we have fixed and that 
        \begin{equation*}
            \disk{\Gamma_H(a)}{\epsilon_a} \cap \convexhull(\disk{\g{m^i_j}}{\epsilon}, \disk{\g{m^{i'}_{j'}}}{\epsilon}) = \emptyset
        \end{equation*}
        for all distinct $i, i' \in [k], j \in [|V^i|], j' \in [|V^{i'}|]$.
        The latter condition serves that we do not intersect with the ``tunnels of visibilies'' discussed in \cref{section:drawing_global_blocking_edges}.
        Finally, we draw each triangle entirely inside $\disk{\Gamma_H(a)}{\epsilon_a}$, ensuring $a$ is drawn on the inner face of each triangle, and each triangle is colored in a distinct color not equal to red. See \cref{figure:choice_gadget} for 
        an example of the local blocking edges associated with the three anchor vertices of a choice gadget drawn inside a choice gadget.

\subsection{Putting it All Together}
\label{section:correctness_proof}

\begin{figure}[t]
    \centering
    \includegraphics[page=7,width=\linewidth]{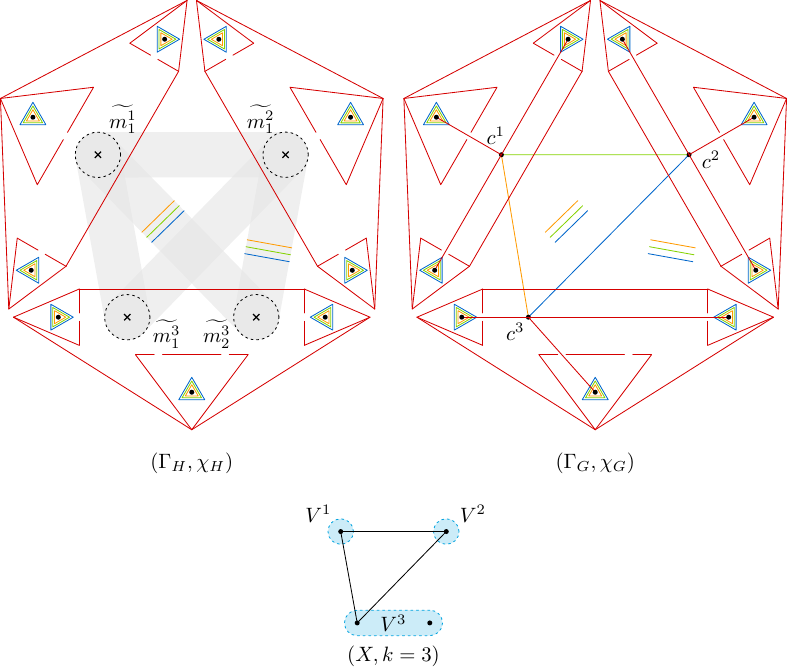}
    \vspace{-0.2cm}\caption{An instance of \textsc{Multicolored Clique} (bottom), the resulting instance of \textsc{GTE}, and a witness showing we have a positive instance (right). 
    Note that the position of anchor vertices and the value of $\epsilon$ are slightly tweaked to increase readability. }
    \label{figure:w1_reduction_example}
\end{figure}

Finally, we are ready to prove the hardness part of \cref{thm:xpwextend}:
Clearly our reduction can be computed in polynomial time.
For a concrete example of the whole reduction, see \cref{figure:w1_reduction_example}.
It remains to ensure our instance mapping is correct:
\begin{lemma}\label{lemma:w1_reduction_correctness}
    The instance $(X, V^1 \cupdot V^2 \cupdot \dots \cupdot V^k, k)$ of \textsc{Multicolored Clique} is positive  if and only if $(H, G, (\Gamma_H, \chi_H), \ell)$ is a positive instance of \textsc{Geometric Thickness Extension}.
\end{lemma}
\begin{proof}
$(\Rightarrow) \colon$
Let $C_X$ be a clique in $X$ where each vertex is in a different $V^i$. We need to construct an extension $(\Gamma_G, \chi_G)$ of $(\Gamma_H, \chi_H)$ that is a \drawing{\ell} of $G$.
For this, it suffices to assign positions to all clique vertices $c^i$ and edge colors to all their incident edges in $G$, and to show that this does not introduce monochromatic crossings nor too many colors.

First, we assign positions:
Let $i \in [k]$ and let $j$ such that $x^i_j \in V(C_X)$.
We set $\Gamma_G(c^i) \coloneqq \g{m^i_j}$.
Next, we assign edge colors:
Color all new edges incident to an anchor vertex red ($=0$), and assign each of the remaining $\binom{k}{2}$ new clique edges a distinct non-red color in $[\binom{k}{2}]$. As $\chi_H$ uses at most colors of $\Set{0} \cup [\binom{k}{2}]$, $\chi_G$ uses at most $\ell$ colors.

We claim that $(\Gamma_G, \chi_G)$ satisfies the conditions of this lemma.
Observe that $\Gamma_G(c^i) \in \g{r_j^i}$.
Hence, by \cref{lemma:anchor_segments}, all red  edges we introduced do not cross a red edge of $(\Gamma_H, \chi_H)$.
Now, consider a new edge $e = c^ic^{i'}$ between two clique vertices. Note that $e$ is non-red, hence it suffices to argue that $e$ crosses neither global nor local blocking edges.
By construction, there are $j, j'$ such that
$x_j^ix_{j'}^{i'} \in E(X)$ and $\Gamma_G(x_{j}^{i}) = \g{m_{j}^{i}} \in \g{r_{j}^{i}}$ and $\Gamma_G(x_{j'}^{i'}) = \g{m_{j'}^{i'}} \in \g{r_{j'}^{i'}}$.
Hence, by \cref{lemma:global_blocking}, $e$ does not cross global blocking edges in $\Gamma_H$.
Since all local blocking edges were positioned to be outside of the convex hull of $\disk{\g{m_{j}^{i}}}{\epsilon}$ and $\disk{\g{m_{j'}^{i'}}}{\epsilon}$,
and by \cref{lemma:regions_in_disks_and_disks_disjoint} both endpoints of $e$ are contained in said hull,
we conclude that $e$ does not cross any local blocking edges.

$(\Leftarrow) \colon$
Let $(\Gamma_G, \chi_G)$ be an extension of $(\Gamma_H, \chi_H)$ that is a \drawing{\ell} of $G$.

Consider a clique vertex $c^i \in C$.
Towards a contradiction, suppose there is no $j \in [|V^i|]$ such that $\Gamma_G(c^i) \in {r_j^i}'$.
Then, by \cref{lemma:anchor_segments}, one of $a^i_lc^i$, $a^i_dc^i$, $a^i_rc^i$ crosses a red segment in $(\Gamma_H, \chi_H)$. Without loss of generality, we assume it is $a^i_lc^i$.
But then, $a^i_lc^i$ must also cross all of the local blocking edges of ${\bl}({a^i_l})$, forcing $\chi_G$ to color this edge red. Hence, $(\Gamma_G, \chi_G)$ contains a monochromatic crossing, a contradiction. Therefore, there is $j \in [|V^i|]$ such that $\Gamma_G(c^i) \in \g{r_j^i}$, and because of \cref{lemma:regions_in_disks_and_disks_disjoint}, $j$ is unique.

We apply this reasoning to all $i \in [k]$ and obtain a function $j \colon [k] \to \bigcup_{i\in[k]} [V^i]$ that selects for each $i \in [k]$ the index of a vertex from $V^i$.
It remains to show that the selected vertices form a clique in $X$.
Let $i, i' \in [k]$ be distinct.
We aim to show that $x^{i}_{j(i)}x^{i'}_{j(i')} \in E(X)$.
Note that by the previous argument $\Gamma_G(c^{i}) \in \g{r^{i}_{j(i)}}$ and $\Gamma_G(c^{i'}) \in \g{r^{i'}_{j(i')}}$.
Observe that each choice gadget in $(\Gamma_G, \chi_G)$ is completely enclosed by red segments, and that all choice gadgets are disjoint.
Hence, $\chi_G$ cannot assign red to $c^ic^{i'}$.
Towards a contradiction, suppose that $x^{i}_{j(i)}x^{i'}_{j(i')} \not\in E(X)$.
Then, by \cref{lemma:global_blocking}, $c^ic^{i'}$ crosses all global blocking edges of this non-edge of $X$.
Hence, $c^ic^{i'}$ cannot have any of the blocking edges colors, and $\chi_G$ must color $c^ic^{i'}$ red, contradicting that $\chi_G$ does not color $c^ic^{i'}$ red.
\qed\end{proof}

\cref{thm:xpwextend} now follows directly from the two proofs at the end of this and the previous subsection.

\xpwextend*

} %

\toappendix{
\section{\NP-Hardness via the Vertex Deletion Distance}
\label{section:gte_by_vertices_para_np}

Finally, we establish \cref{thm:npextend} by showing that \textsc{GTE} parameterized by the number of missing vertices $k$ is \NP-hard, even for $k = 2$.
Similar to \cite{depian2024parameterized}, we reduce from \textsc{3-SAT} while reusing the machinery developed in the previous subsection.

Consider an instance $(X, C)$ of \textsc{3-SAT} where $X = \Set{x_1, \dots, x_{n}}$ is a set of $n$ variables and $C = \Set{c_1, \dots, c_m}$ is a set of $m$ clauses with three literals each. 
 We construct an instance $(H, G, (\Gamma_H, \chi_H), \ell)$ of \textsc{GTE} as follows (see \cref{section:details_paranp} for the precise formal definition).
The target graph $G$ consists of (apart from the gadgets we will describe next) two stars:
One star with center vertex $t$ (the \emph{truth assignment vertex}) and leaves $X$, and one star with center vertex $v$ (the \emph{verification vertex}) and leaves $C$.
The graph $H$ is obtained by deleting $t$ and $v$ from $G$.
In a slight abuse of notation, we use the literals over $X$, i.e.\ $x$ or $\neg x$ for each $x \in X$, as colors for the edges of $G$ and $H$, with ``red'' as one additional color.

Intuitively, we will force both missing vertices to be drawn at two known locations. An edge from the truth assignment vertex to a variable $x$ will only be able to be colored in either $x$ or $\neg x$, inducing a variable assignment.
Due to our construction, whatever color is chosen will not be available for the edges of the verification vertex, where an edge from the verification vertex to a clause $l_1 \lor l_2 \lor l_3$
will need to be colored by one of $\neg l_1, \neg l_2, \neg l_3$, i.e., the negation of a literal that satisfies the clause.
See \cref{figure:paranp} for an example of the reduction.

To this end, we use a square to construct the \drawing{\ell} $(\Gamma_H, \chi_H)$ of $H$.
The vertices $C$ are placed on the square's left side and the vertices $X$ on the square's bottom side. Furthermore, we place a choice gadget (cf.~\cref{section:gte_by_edges_w1_hard}) on the square's top and right side, forcing $v$ and $t$ to be drawn in a small region on the top side ($r_t$) and the right side ($r_v$) respectively.
Note that we use the choice gadgets only to constrain these two vertices to two known positions, not to model a choice.

Finally, we add colorful triangles to $(\Gamma_H, \chi_H)$ around the vertices $X \cup C$ to restrict each $e \in E(G) \setminus E(H)$ in an extension of $(\Gamma_H, \chi_H)$ to a set of \emph{allowed colors} $a(e)$.
For an edge $tx$ with $x \in X$ we set $a(tx) \coloneqq \Set{x, \neg x}$ and 
for an edge $vc$ with $c = l_1 \lor l_2 \lor l_3$, we set $a(vc) \coloneqq \Set{ \neg l_1, \neg l_2, \neg l_3 }$.
The number of allowed colors $\ell$ is set to $2|X|+1$.

\begin{figure}[t]
    \centering
    \includegraphics[page=10]{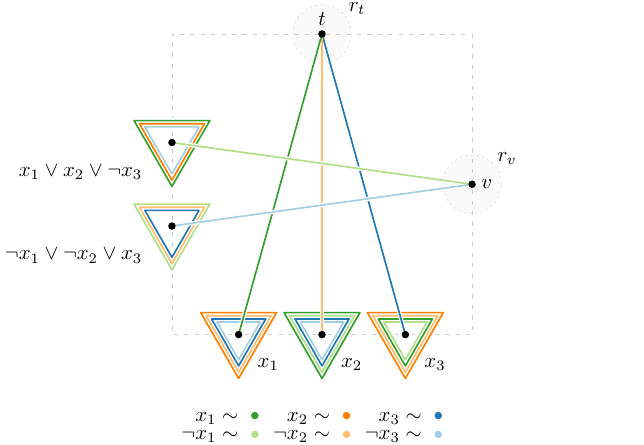}
    \vspace{-0.2cm}\caption{Example of a positive instance of \textsc{GTE} resulting from our reduction from \textsc{3-SAT}. The satisfying assignment is $x_1 \mapsto 1, x_2 \mapsto 0, x_3 \mapsto 1$. Note that the choice gadgets are not fully drawn.}
    \label{figure:paranp}
\end{figure}

The following lemma characterizes the combinatorial properties our construction ensures:
Adding the two missing vertices to $(\Gamma_H, \chi_H)$
yields a \drawing{\ell} iff the two choice gadgets are respected, the allowed colors $a(\cdot)$ are respected, and for each color used in one star, it may not be used in the other:

\begin{restatable}{lemma}{paranpdrawingtocombinatorial}
\label{lemma:para_np_drawing_to_combinatorial}
Let $\Gamma_G$ be an extension of $\Gamma_H$ to $G$ 
and let $\chi_G$ be an extension of $\chi_H$ to $E(G)$.
Then, $(\Gamma_G, \chi_G)$ is a \drawing{\ell} of $G$ if and only if 
    (i) $\Gamma_G(t) \in r_t$ and $\Gamma_G(v) \in r_v$, 
    (ii) $\chi_G(e) \in a(e)$ for all $e \in E(G) \setminus E(H)$, and,
    (iii) $\chi_G(\Set{tx \mid x \in X}) \cap \chi_G(\Set{vc \mid c \in C}) = \emptyset$.
\end{restatable}
\begin{proof}[Sketch]
$(\Rightarrow) \colon$
Assume $(\Gamma_G, \chi_G)$ is a \drawing{\ell} of $G$.
Then, since the local blocking edges of the anchor vertices force the edges of $v$ (resp. $t$) to be red, by sightly adapting \cref{lemma:anchor_segments} to the changed notation, we obtain Property (i).
Property (ii) follows directly from the manner we inserted local blocking edges around the vertices $X \cup C$.
Property (iii) is implied because the edges of $t$ (resp. $v)$ near completely ($\epsilon$ is small enough) ``cut through'' the square vertically (resp. horizontally), hence the ``horizontal visibility'' (resp. ``vertical visibility'') is blocked.

$(\Leftarrow) \colon$
Assume $(\Gamma_G, \chi_H)$ fulfills Properties (i)-(iii).
By construction, no edges of $H$ induce monochromatic crossings.
It thus remains to consider the edge set $E(G) \setminus E(H)$.
Using Property (i), \cref{lemma:anchor_segments}, the fact that each edge gadget is enclosed with a red ``perimeter'', and the construction of the local blocking edges around the anchor vertices, we can deduce that the drawing contains no red crossings.
Property (ii) implies that there are no monochromatic crossings with local blocking edges.
The only monochromatic crossings left to consider are between edges incident to $v$ and edges incident to $t$. But these two edge sets are colored with two disjoint color sets by Property (iii). Hence, $(\Gamma_G, \chi_G)$ is a \drawing{\ell} of $G$.
\qed\end{proof}

Using this lemma, we can prove our final theorem:

\setcounter{oldtheorem}{\value{theorem}}
\setcounter{theorem}{\the\numexpr\getrefnumber{thm:npextend}-1\relax}

\begin{restatable}{theorem}{npextendstar}
\GTE is \NP-hard even if the drawn subgraph can be obtained from the input graph by deleting only two vertices.
\end{restatable}

\setcounter{theorem}{\value{oldtheorem}}

\begin{proof}
The reduction can clearly be computed in polynomial time.
It remains to show its correctness, that is, the instance $(X, C)$ of \textsc{3-SAT} is positive if and only if $(H, G, (\Gamma_H, \chi_H), \ell)$ is a positive instance of \textsc{Geometric Thickness Extension}.

$(\Rightarrow) \colon$
Let $\sigma \colon X \to \Set{0,1}$ be a variable assignment that satisfies $C$.
First, we extend $\Gamma_H$ to a drawing of $G$:
As $r_v, r_t$ are non-empty by construction, we can find $\Gamma_G$ such that $\Gamma_G(v) \in r_v$ and $\Gamma_G(t) \in r_t$.
Second, we extend $\chi_H$ to an edge-coloring of $G$ called $\chi_G$:
Consider the edges of the star with center vertex $t$. Let $x \in X$. We set $\chi_G(tx)$ to $x$ if $\sigma(x) = 1$ and $\neg x$ otherwise.
Next, consider the edges of the star with center vertex $v$. Let $c \in C$ and let $l$ be a literal of $c$ that evaluates to 1 under $\sigma$.
We set $\chi_G(vc)$ to $\neg l$.
Suppose there is a literal $l \in \chi_G(\Set{tx \mid x \in X}) \cap \chi_G(\Set{vc \mid c \in C})$.
Then, by construction, $l$ as well as $\neg l$ evaluate to 1 under $\sigma$, a contradiction.
Hence, we have satisfied the three conditions of \cref{lemma:para_np_drawing_to_combinatorial}, thereby completing our proof.

$(\Leftarrow) \colon$
Let $(\Gamma_G, \chi_G)$ be a \drawing{\ell} of $G$ and let $\sigma \colon X \to \Set{0,1}$ with $x \mapsto 1$  if and only if $\chi_G(tx) = x$.
We claim that $\sigma$ satisfies $C$.
Let $l_1 \lor l_2 \lor l_3 \in C$. By \cref{lemma:para_np_drawing_to_combinatorial}, 
we have $\chi_G(vc) \in a(vc) = \Set{\neg l_1, \neg l_2, \neg l_3 }$. Without loss of generality, assume $\chi_G(vc) = \neg l_1$ and let $x$ be the variable underlying $\neg l_1$.
By \cref{lemma:para_np_drawing_to_combinatorial},
$\chi_G(tx) \in a(tx) = \Set{x, \neg x} = \Set{l_1, \neg l_1}$. Suppose $\chi_G(tx) = \neg l_1$. Then 
$\neg l_1 \in \chi_G(\Set{tx \mid x \in X}) \cap \chi_G(\Set{vc \mid c \in C})$, contradicting \cref{lemma:para_np_drawing_to_combinatorial}.
Hence $\chi_G(tx) = l_1$, implying $l_1$ is true under $\sigma$ and the clause is satisfied.
\qed\end{proof}

\subsection{Details of the \paraNP-Hardness Proof}
\label{section:details_paranp}

In this subsection, we make the construction of $(\Gamma_H, \chi_H)$ from \cref{section:gte_by_vertices_para_np} precise. In particular, we describe how to adapt the construction used in the \W{1} hardness proof 
(\cref{section:gte_by_edges_w1_hard})
 using the outline given in \cref{section:gte_by_vertices_para_np}.
We construct the two choice gadgets with the number of choices set to one at scale 
$s = 1$ and using $\epsilon = 0.5$. 
Instead of constraining the position of clique vertices, one choice gadget constrains the position of $t$ to $r_t$, and an the other constrains the position of $v$ 
 to $r_v$.

The square we use as a template to attach the choice gadgets to has side length $10$, that is, the ``width'' of a choice gadget at scale $s = 1$.
We position the vertices $C$ arbitrarily on the interior of the left side of the square, and do likewise for the vertices $X$ on the bottom side.

We do not insert global blocking edges.
We modify the insertion of local blocking edges as follows:
For each of the 6 anchor vertices of the choice gadget, we insert one colored triangle for all colors besides red.
For each vertex $w \in X \cup C$, we
insert one colored triangle for each of the colors $X \setminus a(w)$.
When placing local blocking edges, instead of avoiding the ``tunnels of visibility'', we avoid the ``cones of visibility'' that arise from each vertex position of $X \cup C$ and the two disks of radius $\epsilon$ enclosing $r_t$ and $r_v$.
} %

\section{Concluding Remarks}
\looseness=1
Our investigation provides the first systematic investigation of the parameterized complexity of computing the geometric thickness of graphs, a fundamental concept which has been studied from a variety of perspectives since its introduction over sixty years ago. The main open questions that arise from our work are resolving the complexity of \GT\ when parameterized by treewidth $\texttt{tw}$ and treedepth $\texttt{td}$---in particular, does there exist a computable function $f$ such that \GT\ can be solved at least in time $|V(G)|^{f(\texttt{tw})}$ or $|V(G)|^{f(\texttt{td})}$?

The main obstacle in the way of obtaining such algorithms seems to be a lack of understanding of solutions for these ``well-structured'' graphs. In particular, if one could show that every $k$-treewidth graph with geometric thickness $\ell$ admits a \drawing{\ell} with some ``suitable combinatorial properties'', this would open the door towards solving \GT\ via the classical dynamic programming approach typically used on graphs of bounded treewidth. On the other hand, if one could show that even bounded-treewidth graphs may require the use of a wide range of ``combinatorially ill-behaved'' \drawing{\ell}s, this would likely open the door to establishing hardness. The issue described above is essentially the reason why several other prominent graph drawing problems remain open when parameterized by treewidth and treedepth, with examples including the computation of \emph{stack}, 
\emph{queue} and \emph{obstacle numbers} of graphs~\cite{BhoreGMN20,GanianMNZ21,Zehavi22,BalkoCG00V022,BhoreGMN22}; the first of these is in fact equivalent to the variant of \GT\ where vertices are forced to lie in convex position.
To make this issue more concrete, we remark that the staple approach for finding a treedepth-kernel (see e.g. \cite{bannister2013parameterized}), that is, contracting repeated subtrees with a similar structure into single vertices, cannot be applied in a straightforward manner, as e.g., there may be unrelated interfering edges passing through the to-be contracted subgraph in a given drawing.

As an alternative approach to identify meaningful tractable fragments, we investigated \GT in the extension setting.
Surprisingly, the problem turned out to be \NP-complete even in the case where only two vertices are missing from a given solution.
For \GTE, future work could target the question of whether there are natural circumstances under which one can achieve fixed-parameter tractability even if vertices are missing from the partial drawing---for instance, is the problem \FPT\ when parameterized by the \emph{vertex+edge deletion distance}~\cite{EibenGHKN20,EibenGHKN20b,depian2024parameterized} plus the number of layers?

\bibliographystyle{splncs04}
\bibliography{refs}

\ifshort
\setbool{inappendix}{true}
\appendix
\clearpage

\section*{Appendix}
\appendixText
\fi

\end{document}